\documentclass[a4paper,UKenglish,cleveref,hyperref,autoref]{lipics-v2021}

\title{Deciding twin-width at most 4 is NP-complete}
\titlerunning{Deciding twin-width at most 4 is NP-complete}

\author{Pierre Bergé}{Univ Lyon, CNRS, ENS de Lyon, Université Claude Bernard Lyon 1, LIP UMR5668, France}{}{}{}
\author{\'{E}douard Bonnet}{Univ Lyon, CNRS, ENS de Lyon, Université Claude Bernard Lyon 1, LIP UMR5668, France \and \url{http://perso.ens-lyon.fr/edouard.bonnet/}}{edouard.bonnet@ens-lyon.fr}{https://orcid.org/0000-0002-1653-5822}{}
\author{Hugues Déprés}{Univ Lyon, CNRS, ENS de Lyon, Université Claude Bernard Lyon 1, LIP UMR5668, France \and \url{http://perso.ens-lyon.fr/hugues.depres/}}{hugues.depres@ens-lyon.fr}{}{}

\authorrunning{P. Bergé, \'E. Bonnet, H. Déprés}

\Copyright{Pierre Bergé, Édouard Bonnet, Hugues Déprés}



\category{}

\relatedversion{}

\supplement{}

\acknowledgements{}

\nolinenumbers 

\hideLIPIcs  

\EventEditors{John Q. Open and Joan R. Access}
\EventNoEds{2}
\EventLongTitle{42nd Conference on Very Important Topics (CVIT 2016)}
\EventShortTitle{CVIT 2016}
\EventAcronym{CVIT}
\EventYear{2016}
\EventDate{December 24--27, 2016}
\EventLocation{Little Whinging, United Kingdom}
\EventLogo{}
\SeriesVolume{42}
\ArticleNo{23}

\usepackage[utf8]{inputenc}  

\usepackage[T1]{fontenc}
\usepackage{lmodern}

\usepackage[colorinlistoftodos,bordercolor=orange,backgroundcolor=orange!20,linecolor=orange,textsize=normalsize]{todonotes}

\usepackage{amsmath}  
\usepackage{amssymb}     
\usepackage{bbm}
\usepackage{accents}
\usepackage{complexity}

\usepackage{booktabs}
\usepackage{hyperref}
\usepackage{paralist}
\usepackage{fixmath}

\crefformat{equation}{\textup{#2(#1)#3}}
\crefrangeformat{equation}{\textup{#3(#1)#4--#5(#2)#6}}
\crefmultiformat{equation}{\textup{#2(#1)#3}}{ and \textup{#2(#1)#3}}
{, \textup{#2(#1)#3}}{, and \textup{#2(#1)#3}}
\crefrangemultiformat{equation}{\textup{#3(#1)#4--#5(#2)#6}}%
{ and \textup{#3(#1)#4--#5(#2)#6}}{, \textup{#3(#1)#4--#5(#2)#6}}{, and \textup{#3(#1)#4--#5(#2)#6}}

\Crefformat{equation}{#2Equation~\textup{(#1)}#3}
\Crefrangeformat{equation}{Equations~\textup{#3(#1)#4--#5(#2)#6}}
\Crefmultiformat{equation}{Equations~\textup{#2(#1)#3}}{ and \textup{#2(#1)#3}}
{, \textup{#2(#1)#3}}{, and \textup{#2(#1)#3}}
\Crefrangemultiformat{equation}{Equations~\textup{#3(#1)#4--#5(#2)#6}}%
{ and \textup{#3(#1)#4--#5(#2)#6}}{, \textup{#3(#1)#4--#5(#2)#6}}{, and \textup{#3(#1)#4--#5(#2)#6}}

\makeatletter
\newtheorem*{rep@theorem}{\rep@title}
\newcommand{\newreptheorem}[2]{%
\newenvironment{rep#1}[1]{%
 \def\rep@title{#2 \ref{##1}}%
 \begin{rep@theorem}}%
 {\end{rep@theorem}}}
\makeatother

\newreptheorem{theorem}{Theorem}
\newreptheorem{lemma}{Lemma}

\usepackage{xspace}
\usepackage{tikz}

\usepackage[ruled,vlined,linesnumbered]{algorithm2e}

\usetikzlibrary{fit}
\usetikzlibrary{arrows}
\usetikzlibrary{patterns}
\usetikzlibrary{calc}
\usetikzlibrary{shapes}
\usetikzlibrary{positioning}
\usetikzlibrary{math}
\usetikzlibrary{shapes.symbols}
\usetikzlibrary{decorations.pathreplacing}
\tikzset{draw half paths/.style 2 args={%
  decoration={show path construction,
    lineto code={
      \draw [#1] (\tikzinputsegmentfirst) -- 
         ($(\tikzinputsegmentfirst)!0.5!(\tikzinputsegmentlast)$);
      \draw [#2] ($(\tikzinputsegmentfirst)!0.5!(\tikzinputsegmentlast)$)
        -- (\tikzinputsegmentlast);
    }
  }, decorate
}}

\usepackage[scr=boondox,scrscaled=1.05]{mathalfa}

\renewcommand{\geq}{\geqslant}
\renewcommand{\leq}{\leqslant}

\renewcommand{\le}{\leq}
\renewcommand{\ge}{\geq}

\newcommand{\card}[1]{|{#1}|}

\theoremstyle{definition}

\newcommand{\tww}{tww}

\colorlet{npink}{red!30!pink}

\begin{document}

\maketitle

\begin{abstract}
  We show that determining if an $n$-vertex graph has twin-width at most~4 is NP-complete, and requires time $2^{\Omega(n/\log n)}$ unless the Exponential-Time Hypothesis fails.
  Along the way, we give an elementary proof that $n$-vertex graphs subdivided at least $2 \log n $ times have twin-width at most~4.
  We also show how to encode trigraphs $H$ (2-edge colored graphs involved in the definition of twin-width) into graphs $G$, in the sense that every $d$-sequence (sequence of vertex contractions witnessing that the twin-width is at most~$d$) of $G$ inevitably creates $H$ as an induced subtrigraph, whereas there exists a partial $d$-sequence that actually goes from $G$ to $H$.
  We believe that these facts and their proofs can be of independent interest.
\end{abstract}
\maketitle

\section{Introduction}\label{sec:intro}

A~\emph{trigraph} is a graph with some edges colored black, and some colored red.
A~(vertex) \emph{contraction} consists of merging two (non-necessarily adjacent) vertices, say, $u, v$ into a~vertex~$w$, and keeping every edge $wz$ black if and only if $uz$ and $vz$ were previously black edges.
The other edges incident to $w$ become red (if not already), and the rest of the trigraph stays the same.
A~\emph{contraction sequence} of an $n$-vertex graph $G$ is a sequence of trigraphs $G=G_n, \ldots, G_1=K_1$ such that $G_i$ is obtained from $G_{i+1}$ by performing one contraction.
A~\mbox{\emph{$d$-sequence}} is a contraction sequence where all the trigraphs have red degree at most~$d$.
The~\emph{twin-width} of $G$, denoted by $\tww(G)$, is then the minimum integer~$d$ such that $G$ admits a $d$-sequence.
See~\cref{fig:contraction-sequence} for an example of a graph admitting a 2-sequence.
The~\emph{red graph} of a trigraph is obtained by simply deleting its black edges.
A~\emph{partial $d$-sequence} is similar to a $d$-sequence but ends on any trigraph $G_i$, instead of on the 1-vertex (tri)graph $G_1$.  
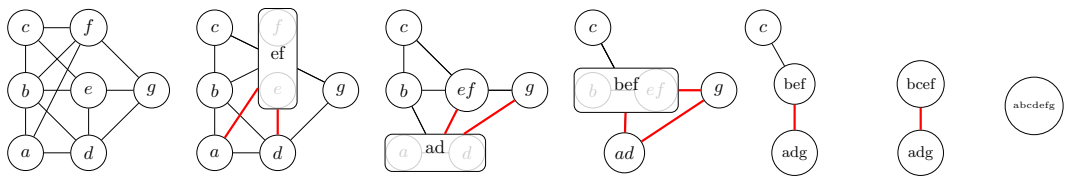
\begin{figure}
  \centering
  \resizebox{400pt}{!}{
  \begin{tikzpicture}[
      vertex/.style={circle, draw, minimum size=0.68cm}
    ]
    \def\s{1.2}
    \foreach \i/\j/\l in {0/0/a,0/1/b,0/2/c,1/0/d,1/1/e,1/2/f,2/1/g}{
      \node[vertex] (\l) at (\i * \s,\j * \s) {$\l$} ;
    }
    \foreach \i/\j in {a/b,a/d,a/f,b/c,b/d,b/e,b/f,c/e,c/f,d/e,d/g,e/g,f/g}{
      \draw (\i) -- (\j) ;
    }

    \begin{scope}[xshift=3 * \s cm]
    \foreach \i/\j/\l in {0/0/a,0/1/b,0/2/c,1/0/d,2/1/g}{
      \node[vertex] (\l) at (\i * \s,\j * \s) {$\l$} ;
    }
    \foreach \i/\j/\l in {1/1/e,1/2/f}{
      \node[vertex,opacity=0.2] (\l) at (\i * \s,\j * \s) {$\l$} ;
    }
    \node[draw,rounded corners,inner sep=0.01cm,fit=(e) (f)] (ef) {ef} ;
    \foreach \i/\j in {a/b,a/d,b/c,b/d,b/ef,c/ef,c/ef,d/g,ef/g,ef/g}{
      \draw (\i) -- (\j) ;
    }
    \foreach \i/\j in {a/ef,d/ef}{
      \draw[red, very thick] (\i) -- (\j) ;
    }
    \end{scope}

    \begin{scope}[xshift=6 * \s cm]
    \foreach \i/\j/\l in {0/1/b,0/2/c,2/1/g,1/1/ef}{
      \node[vertex] (\l) at (\i * \s,\j * \s) {$\l$} ;
    }
    \foreach \i/\j/\l in {0/0/a,1/0/d}{
      \node[vertex,opacity=0.2] (\l) at (\i * \s,\j * \s) {$\l$} ;
    }
    \draw[opacity=0.2] (a) -- (d) ;
    \node[draw,rounded corners,inner sep=0.01cm,fit=(a) (d)] (ad) {ad} ;
    \foreach \i/\j in {ad/b,b/c,b/ad,b/ef,c/ef,c/ef,ef/g,ef/g}{
      \draw (\i) -- (\j) ;
    }
    \foreach \i/\j in {ad/ef,ad/g}{
      \draw[red, very thick] (\i) -- (\j) ;
    }
    \end{scope}

    \begin{scope}[xshift=9 * \s cm]
    \foreach \i/\j/\l in {0/2/c,2/1/g,0.5/0/ad}{
      \node[vertex] (\l) at (\i * \s,\j * \s) {$\l$} ;
    }
    \foreach \i/\j/\l in {0/1/b,1/1/ef}{
      \node[vertex,opacity=0.2] (\l) at (\i * \s,\j * \s) {$\l$} ;
    }
    \draw[opacity=0.2] (b) -- (ef) ;
    \node[draw,rounded corners,inner sep=0.01cm,fit=(b) (ef)] (bef) {bef} ;
    \foreach \i/\j in {ad/bef,bef/c,bef/ad,c/bef,c/bef,bef/g}{
      \draw (\i) -- (\j) ;
    }
    \foreach \i/\j in {ad/bef,ad/g,bef/g}{
      \draw[red, very thick] (\i) -- (\j) ;
    }
    \end{scope}

    \begin{scope}[xshift=11.7 * \s cm]
    \foreach \i/\j/\l in {0/2/c}{
      \node[vertex] (\l) at (\i * \s,\j * \s) {$\l$} ;
    }
     \foreach \i/\j/\l in {0.5/0/adg,0.5/1.1/bef}{
      \node[vertex] (\l) at (\i * \s,\j * \s) {\footnotesize{\l}} ;
    }
    \foreach \i/\j in {c/bef}{
      \draw (\i) -- (\j) ;
    }
    \foreach \i/\j in {adg/bef}{
      \draw[red, very thick] (\i) -- (\j) ;
    }
    \end{scope}

    \begin{scope}[xshift=13.7 * \s cm]
    \foreach \i/\j/\l in {0.5/0/adg,0.5/1.1/bcef}{
      \node[vertex] (\l) at (\i * \s,\j * \s) {\footnotesize{\l}} ;
    }
    \foreach \i/\j in {adg/bcef}{
      \draw[red, very thick] (\i) -- (\j) ;
    }
    \end{scope}

    \begin{scope}[xshift=15 * \s cm]
    \foreach \i/\j/\l in {1/0.75/abcdefg}{
      \node[vertex] (\l) at (\i * \s,\j * \s) {\tiny{\l}} ;
    }
    \end{scope}
    
  \end{tikzpicture}
  }
  \caption{A 2-sequence witnessing that the initial graph has twin-width at most~2.}
  \label{fig:contraction-sequence}
\end{figure}
Twin-width can be naturally extended to matrices over a finite alphabet (in an unordered~\cite{twin-width1}, or an ordered setting~\cite{twin-width4}), and hence to any binary structure. 

Surprisingly many classes turn out to be of bounded twin-width.
Such is the case of graphs with bounded clique-width, $H$-minor free graphs for any fixed $H$, posets with antichains of bounded size, strict subclasses of permutation graphs, map graphs, bounded-degree string graphs~\cite{twin-width1}, as well as $\Omega(\log n)$-subdivisions of $n$-vertex graphs, and some classes of cubic expanders~\cite{twin-width2}.
One of the main algorithmic interests with twin-width is that first-order (FO) model checking, that is, deciding if a sentence $\varphi$ holds in a graph $G$, can be decided in fixed-parameter time (FPT) $f(|\varphi|,d) \cdot |V(G)|$ for some computable function $f$, when given a $d$-sequence of $G$~\cite{twin-width1}. 
As for most classes known to have bounded twin-width, one can compute $O(1)$-sequences in polynomial time for members of the class, the latter result unifies and extends several known results~\cite{Flum01,Gajarsky14,Ganian15,Gajarsky15,Guillemot14} for hereditary (but not necessarily monotone) classes.

For monotone (i.e., subgraph-closed) classes, the FPT algorithm of Grohe, Kreutzer, and Siebertz~\cite{Grohe17} for FO model checking on nowhere dense classes, is complemented by W$[1]$-hardness on classes that are somewhere dense (i.e., \emph{not} nowhere dense)~\cite{Dvorak13}, and even AW$[*]$-hardness on classes that are \emph{effectively} somewhere dense~\cite{Kreutzer09}.
The latter results mean that, for monotone classes, FO model checking is unlikely to be FPT beyond nowhere dense classes.

The missing piece for an FO model-checking algorithm in FPT time on any class of bounded twin-width is a polynomial-time algorithm and a computable function $f$, that given a constant integer bound $c$ and a graph $G$, either finds an $f(c)$-sequence for $G$, or correctly reports that the twin-width of $G$ is greater than $c$.
The running time of the algorithm could be $n^{g(c)}$, for some function $g$.
However to get an FPT algorithm in the combined parameter \emph{size of the sentence + bound on the twin-width}, one would further require that the approximation algorithm takes FPT time in $c$ (now seen as a parameter), i.e., $g(c) n^{O(1)}$.
We know such an algorithm for instance on ordered graphs (more generally, ordered binary structures)~\cite{twin-width4}, graphs of bounded clique-width, proper minor-closed classes~\cite{twin-width1}, but not in general graphs.

On the other hand, prior to this paper, no algorithmic lower bound was known for computing the twin-width.
Our main result rules out an (exact) XP algorithm to decide \emph{$\tww(G) \leqslant k$}, that is, an algorithm running in time $n^{f(k)}$ for some computable function $f$.
Indeed we show that deciding if the twin-width of a graph is at most~4 is intractable.
We refer the reader to~\cref{sec:prelim} for some context on the Exponential-Time Hypothesis (ETH), which implies that $n$-variable \textsc{$3$-SAT} cannot be solved in time $2^{o(n)}$. 

\begin{theorem}\label{thm:main}
  Deciding if a graph has twin-width at most~4 is NP-complete.
  Furthermore, no algorithm running in time $2^{o(n/\log n)}$ can decide if an $n$-vertex graph has twin-width at most~4, unless the ETH fails.
\end{theorem}

As far as approximation algorithms are concerned, our result only rules out a ratio better than $5/4$ for determining the twin-width.
This still leaves plenty of room for an $f(\text{OPT})$-approximation, which would be good enough for most of the (theoretical) algorithmic applications.
Note that such algorithms exist for treewidth in polytime~\cite{Feige08} and FPT time~\cite{Korhonen21}, for pathwidth~\cite{Feige08}, and for clique-width via rank-width~\cite{Oum08}.

Is~\cref{thm:main} surprising?
On the one hand, it had to be expected that deciding, given a graph $G$ and an integer $k$, whether $\tww(G) \leqslant k$ would be NP-complete.
This is the case for example of treewidth~\cite{Arnborg87}, pathwidth~\cite{Ohtsuki79,Kashiwabara79,Lengauer81}, clique-width~\cite{Fellows09}, rank-width~\cite{Hlineny08}, mim-width~\cite{Saether16}, and bandwidth~\cite{Garey79}.
On the other hand, the parameterized complexity of these width parameters is more diverse and harder to predict.
Famously, Bodlaender's algorithm is a linear FPT algorithm to exactly compute treewidth~\cite{Bodlaender96} (and a non-uniform FPT algorithm came from the Graph Minor series~\cite{Robertson95}).
In contrast, while there is an XP algorithm to compute bandwidth~\cite{Saxe80}, an FPT algorithm is highly unlikely~\cite{Dregi14}.
It is a long-standing open whether an FPT or a mere XP algorithm exist for computing clique-width exactly, or even simply if one can recognize graphs of clique-width at most~4 in polynomial time (deciding clique-width at most~3 is indeed tractable~\cite{Corneil12}).

\Cref{thm:main} almost completely resolves the parameterized complexity of exactly computing twin-width on general graphs.
Two questions remain: can graphs of twin-width at most~2, respectively at most~3, be recognized in polynomial time.
Graphs of twin-width~0 are cographs, which can be recognized in linear time~\cite{Habib05}, while it was recently shown that graphs of twin-width at most~1 can be recognized in polynomial time~\cite{tww-polyker}.
In the course of establishing~\cref{thm:main} we show and generalize the following, where an \emph{$(\geqslant s)$-subdivision} of a graph is obtained by subdividing each of its edges at least $s$ times.
\begin{theorem}\label{thm:long-subd-gr}
  Any $(\geqslant 2 \log n)$-subdivision of an $n$-vertex graph has twin-width at most~4.
\end{theorem}
That those graphs have bounded twin-width was known~\cite{twin-width2}, but not with the explicit bound of~4.
Another family of graphs with twin-width at most~4 is the set of grids, walls, their subgraphs, and subdivisions.
Even if there is no proof of that fact, sufficiently large grids, walls, or their subdivisions likely have twin-width \emph{at least}~4; it is actually surprising that the $6 \times 8$ grid still has twin-width~3~\cite{Schidler21}.
We also believe that long subdivisions of ``sufficiently complicated'' graphs have twin-width \emph{at least}~4.
That would make graphs of twin-width at most~3 considerably simpler than those of twin-width at most~4, especially among sparse graphs.

Contrary to the hardness proof for treewidth~\cite{Arnborg87}, which involves some structural characterizations by chordal completions, and the intermediate problems \textsc{Minimum Cut Linear Arrangement}, \textsc{Max Cut}, and \textsc{Max $2$-SAT} \cite{Garey79}, our reduction is ``direct'' from \textsc{$3$-SAT}.
This makes the proven hardness of twin-width more robust, and easier to extend to restricted classes of graphs, especially sparse ones.
\Cref{thm:main} holds for bounded-degree input graphs.
For instance, performing our reduction from~\textsc{Planar $3$-SAT} produces subgraphs of constant powers of the planar grid (while admittedly weakening the ETH lower bound from $2^{\Omega(n/\log n)}$ to $2^{\Omega(\sqrt n/\log n)}$).
Hence, while the complexity status of computing treewidth on planar graphs is a famous long-standing open question, one can probably extend the NP-hardness of \emph{twin-width at most~4} to planar graphs, by tuning and/or replacing the few non-planar gadgets of our reduction. 

Let us point out that, in contrast to subset problems, there is no $2^{O(n)}$-time algorithm known to compute twin-width.
The exhaustive search takes time $n^{2n+O(1)}$ by considering all sequences of $n-1$ pairs of vertices.
We leave as an open question whether the ETH lower bound of computing twin-width can be brought from $2^{\Omega(n/\log n)}$ to $2^{\Omega(n)}$, or even $2^{\Omega(n \log n)}$.
The latter lower bound is known to hold for \textsc{Subgraph Isomorphism}~\cite{Cygan17} (precisely, given a graph $H$ and an $n$-vertex graph $G$, deciding if $H$ is isomorphic to a subgraph of $G$ requires time $2^{\Omega(n \log n)}$, unless the ETH fails), or computing the Hadwiger number~\cite{Fomin21} (i.e., the size of the largest clique minor).

\subsection{Outline of the proof of~\cref{thm:main}}

We propose a quasilinear reduction from \textsc{$3$-SAT}.
Given an $n$-variable instance $I$ of \textsc{$3$-SAT}, we shall construct an $O(n \log n)$-vertex graph $G=G(I)$ which has twin-width at most~4 if and only if $I$ is satisfiable.

Half of our task is to ensure that no 4-sequence will exist if $I$ is unsatisfiable.
This is challenging since many contraction strategies are to be considered and addressed.
We make this task more tractable by attaching \emph{fence gadgets} to some chosen vertex subsets.
The effect of the fence \emph{enclosing}~$S$ is that no contraction can involve vertices in $S$ with vertices outside of~$S$, while~$S$ is not contracted into a single vertex.
The \emph{maximal or outermost} fences (we may nest two or more fence gadgets) partition the rest of the vertices.
This significantly tames the potential 4-sequences of~$G$.

Our basic building block, the \emph{vertical set}, consists of a pair of vertices (\emph{vertical pair}) enclosed by a fence.
It can be thought of as a bit set to 0 as long as the pair is not contracted, and to 1 when the pair gets contracted.
It is easy to assemble vertical sets as prescribed by an auxiliary digraph $D$ (of maximum degree~3), in such a way that, to contract (by a partial 4-sequence) the pair of a vertical set $V$, one first has to contract all the vertical sets that can reach $V$ in $D$.
This allows to propagate and duplicate a bit in a so-called \emph{wire} (corresponding to an out-tree in $D$), and to perform the logical AND of two bits.

The bit propagation originates from a variable gadget (we naturally have one per variable appearing in $I$) that offers two alternatives.
One can contract the ``top half'' of the gadget of variable $x_i$, which then lets one contract the vertical sets in the wire of literal $x_i$, or one can contract instead the ``bottom half'' of the gadget, as well as the vertical sets in the wire of literal $\neg x_i$. 
Concretely, these two contraction schemes represent the two possible assignments for variable $x_i$. 
A special ``lock'' on the variable gadget (called \emph{half-guards}) prevents its complete contraction, and in particular, performing contractions in both the wires of a literal and its negation.

The leaves of the literal wires serve as inputs for 3-clause gadgets.
One can contract the output (also a vertical set) of a clause gadget if and only if one of its input is previously contracted.
We then progressively make the AND of the clauses via a ``path'' of binary AND gadgets fed by the clause outputs.
We eventually get a vertical set, called \emph{global output}, which can be contracted by a partial 4-sequence only if $I$ is satisfiable.
Indeed at this point, the variable gadgets are still locked so at most one of their literals can be propagated.
This ticks one of our objective off.
We should now ensure that a 4-sequence is possible from there, when $I$ is satisfiable.

For that purpose, we add a wire from the global output back to the half-guards (or locks) of the variable gadgets.
One can contract the vertical sets of that wire, and in particular the half-guards.
Once the variable gadgets are ``unlocked,'' they can be fully contracted.
As a consequence, one can next contract the wires of literals set to false, and \emph{all} the remaining vertical sets involved in clause gadgets.

At this point, the current trigraph $H$ roughly has one vertex per outermost fence with red edges linking two adjacent gadgets (and no black edge).
We will guarantee that the (red) degree of $H$ is at most~4, its number of vertices of degree at least 3 is at most $\beta n$, for some constant $\beta$.
Besides we will separate gadgets by degree-2 wires of length $2 \log(\beta n)$ beforehand.
This is crucial so that the red graph of $H$ is a $(2 \log n')$-subdivision of an $n'$-vertex graph.
We indeed show that such trigraphs have twin-width at most~4.
A complicated proof in~\cite{twin-width2} shows that $\Theta(\log n')$-subdivisions of $n'$-vertex graphs have bounded twin-width.
Here we give an elementary proof of a similar fact with an explicit upper bound of~4.

This finishes to describe our overall plan for the reduction and its correctness.
It happens that fence gadgets are easier to build as trigraphs, while the rest of the gadgetry can be directly encoded by graphs.
We thus show how to encode trigraphs by graphs, as follows.
For any trigraph $J$ whose red graph has degree at most $d$, and component size at most $h$, there is a graph $G$ on at most $f(d,h) \cdot |V(J)|$ vertices such that $J$ has twin-width at most $2d$ if and only if $G$ has twin-width at most $2d$.
This uses some local replacements and confluence properties of certain partial contraction sequences.

\subsection{Organization of the paper}

The rest of the paper is organized as follows.
In~\cref{sec:prelim} we formally introduce the relevant definitions and notations, and state some easy observations. 
In~\cref{sec:long-subd} we show a~generalization to trigraphs of~\cref{thm:long-subd-gr}, crucial to the subsequent hardness proof.
\Cref{sec:hardness} is devoted to the hardness construction.
In \cref{subsec:trigraph-encoding} we present how to encode trigraphs into graphs.
In \cref{subsec:foreword}, we quickly recap the overall plan.
The subsequent subsections go through the various gadgets.
Finally \cref{subsec:constr-corr} details the quasilinear reduction from \textsc{$3$-SAT} to the problem of deciding if the twin-width is at most~4, and its correctness.

\section{Preliminaries}\label{sec:prelim}

For $i$ and $j$ two integers, we denote by $[i,j]$ the set of integers that are at least $i$ and at most $j$.
For every integer $i$, $[i]$ is a shorthand for $[1,i]$.
We use the standard graph-theoretic notations: $V(G)$ denotes the vertex set of a graph $G$, $E(G)$ denotes its edge set, $G[S]$ denotes the subgraph of $G$ induced by $S$, etc.
An $(\geqslant s)$-subdivision (resp.~$s$-subdivision) of a graph $G$ is obtained by subdividing every edge of $G$ at least $s$ times (resp.~exactly $s$ times).

\subsection{Definitions and notations related to twin-width}

We first define partition sequences, which is an alternative approach to contraction sequences.

\medskip

\textbf{Partition sequences.} The \emph{twin-width} of a graph, introduced in~\cite{twin-width1}, can be defined in the following way (complementary to the one given in introduction).
A~\emph{partition sequence} of an $n$-vertex graph $G$, is a sequence $\mathcal P_n, \ldots, \mathcal P_1$ of partitions of its vertex set $V(G)$, such that $\mathcal P_n$ is the set of singletons $\{\{v\}~:~v \in V(G)\}$, $\mathcal P_1$ is the singleton set $\{V(G)\}$, and for every $2 \leqslant i \leqslant n$, $\mathcal P_{i-1}$ is obtained from $\mathcal P_i$ by merging two of its parts into one.
Two parts $P, P'$ of a same partition $\mathcal P$ of $V(G)$ are said \emph{homogeneous} if either every pair of vertices $u \in P, v \in P'$ are non-adjacent, or every pair of vertices $u \in P, v \in P'$ are adjacent.
Finally the twin-width of $G$, denoted by $\tww(G)$, is the least integer $d$ such that there is partition sequence $\mathcal P_n, \ldots, \mathcal P_1$ of $G$ with every part of every $\mathcal P_i$ ($1 \leqslant i \leqslant n$) being homogeneous to every other parts of $\mathcal P_i$ but at most~$d$.

\medskip

\textbf{Contraction sequences.} 
A~\emph{trigraph} $G$ has vertex set $V(G)$, black edge set $E(G)$, red edge set $R(G)$.
Its~\emph{red graph} $(V(G),R(G))$ may be denoted $\mathcal R(G)$, and \emph{total graph} $(V(G),E(G) \cup R(G))$, $\mathcal T(G)$.
The \emph{red degree} of a trigraph is the degree of its red graph.
A trigraph $G'$ is an \emph{induced subtrigraph} of trigraph $G$ if $V(G') \subseteq V(G)$, $E(G')=E(G) \cap {V(G') \choose 2}$, and $R(G')=R(G) \cap {V(G') \choose 2}$.
Then we say that $G$ is a supertrigraph of $G'$, and we may also denote $G'$ by $G[V(G')]$.
The definition of the previous paragraph is equivalent to the one given in introduction, via contraction sequences.
Indeed the trigraph $G_i$ is obtained from partition $\mathcal P_i$, by having one vertex per part of $\mathcal P_i$, a black edge between any fully adjacent pair of parts, and a red edge between any non-homogeneous pair of parts.
A \emph{partial contraction sequence} is then a sequence of trigraphs of $\mathcal P_n, \ldots, \mathcal P_i$ (for some $i \in [n]$).
A (non-partial) \emph{contraction sequence} is one such that $i=1$.
Note that going from the trigraph of $\mathcal P_{i+1}$ to the one of $\mathcal P_i$ corresponds to the contraction operation described in~\cref{sec:intro}.
A \emph{(partial) $d$-sequence} is a (partial) contraction sequence where all the trigraphs have red degree at most~$d$.
In the above definitions of twin-width, nothing prevents the starting structure $G$ to be a trigraph itself.
We may then talk about the twin-width of a trigraph.

\medskip

\textbf{Conversion partition/contraction, and convenient notations and definitions.}
The reason we gave two (equivalent) definitions of twin-width is that both viewpoints are incomparably useful and convenient.
It is always cumbersome to describe a sequence of partitions, so we will prefer the trigraph and contraction viewpoint to describe $d$-sequences.
However the trigraph ``loses'' a part of the information that the partitioned graph preserves; namely the exact adjacencies between non-homogeneous parts, and the names of the original vertices.
To navigate between these two worlds, and keep the proofs compact, we use the following notations and vocabulary.

In this paragraph, we assume that there is a partial contraction sequence from (tri)graph $G$ to trigraph $H$.
If $u$ is a vertex of $H$, then $u(G)$ denotes the set of vertices eventually contracted into $u$ in $H$.
We denote by $\mathcal P(H)$ the partition $\{u(G) : u \in V(H)\}$ of $V(G)$.
If $G$ is clear from the context, we may refer to a \emph{part} of $H$ as any set in $\{u(G) : u \in V(H)\}$.
We may say that two parts $y(G), z(G)$ of $\mathcal P(H)$ are \emph{in conflict} if $yz \in R(H)$.
We say that a contraction of two vertices $u, u' \in V(H)$ \emph{involves} a vertex $v \in V(G)$ if $v \in u(G)$ or $v \in u'(G)$.
A contraction \emph{involves} a pair of vertices $v, v'$ if $v \in u(G)$ and $v' \in u'(G)$ (or $v \in u'(G)$ and $v' \in u(G)$).
Observe that the two vertices should appear in two distinct parts.
By extension, we may say that a contraction \emph{involves} a set $S$, if it involves a vertex of $S$, or a pair of sets $S, T$ if it involves a pair in $S \times T$.

\subsection{Useful observations}

The twin-width can only decrease when taking induced subtrigraph.
\begin{observation}\label{obs:subtrigraph}
  Let $G'$ be an induced subtrigraph of trigraph $G$.
  Then $\tww(G') \leqslant \tww(G)$.
\end{observation}

It can be observed that adding red edges can only increase the twin-width, since the same contraction sequence for the resulting trigraph works at least as well for the initial trigraph.
\begin{observation}\label{obs:adding-red-edges}
  Let $G$ be a trigraph and $G'$ another trigraph obtained from $G$ by turning some non-edges and some edges into red edges.
  Then $\tww(G') \geqslant \tww(G)$.
\end{observation}

Trees admit a simple 2-sequence, that gives a $d$-sequence on red trees of degree at most~$d$.
\begin{lemma}[\cite{twin-width1}]\label{lem:tree}
  Every (black) tree has twin-width at most~2.
  Every red tree has twin-width at most its maximum degree.
\end{lemma}
\begin{proof}
  Root the tree arbitrarily.
  Contract two leaves with the same parent whenever possible.
  If not possible, contract a deepest leaf with its parent.
  Observe that this cannot result in red degree larger than the maximum of 2 and the initial degree of the tree.
\end{proof}

\subsection{The Exponential-Time Hypothesis}

The Exponential-Time Hypothesis (ETH, for short) was proposed by Impagliazzo and Paturi~\cite{Impagliazzo01} and asserts that there is no subexponential-time algorithm solving \textsc{$3$-SAT}.
More precisely, there is an $\varepsilon > 0$ such that $n$-variable \textsc{$3$-SAT} cannot be solved in time $2^{\varepsilon n}$.
Impagliazzo et al.~\cite{sparsification} present a subexponential-time Turing-reduction parameterized by a~positive real $\varepsilon > 0$ that, given an $n$-variable $m$-clause \textsc{$k$-SAT}-instance $I$, produces at most $2^{\varepsilon n}$ \textsc{$k$-SAT}-instances $I_1, \ldots, I_t$ such that $I$ is satisfiable if and only if at least one of $I_1, \ldots, I_t$ is satisfiable, each $I_i$ having no more than $n$ variables and $C_\varepsilon n$ clauses for some constant~$C_\varepsilon$ depending only on $\varepsilon$.
This crucial reduction is known as the Sparsification Lemma.
Due to it, there is an $\varepsilon' > 0$ such that there is no algorithm solving $m$-clause \textsc{$3$-SAT} in time $2^{\varepsilon' m}$, unless the ETH fails.

A classic reduction from Tovey~\cite{Tovey84}, linear in the number of clauses, then shows the following.
\begin{theorem}[\cite{Tovey84,sparsification}]\label{thm:3sat4occ}
 The $n$-variable \textsc{$3$-SAT} problem where each variable appears at most twice positively, and at most twice negatively, is NP-complete, and cannot be solved in time $2^{o(n)}$, unless the ETH fails. 
\end{theorem}

\section{Long subdivisions have twin-width at most four}\label{sec:long-subd}

In~\cite{twin-width2}, it is proved that the $\Omega(\log n)$-subdivision of any $n$-vertex graph has bounded twin-width.
The proof is rather involved, relies on a characterization by \emph{mixed minors} established in~\cite{twin-width1}, and does not give an explicit constant bound.
Here we give an elementary proof that any $(\geqslant 2 \lceil \log n \rceil - 1)$-subdivision of an $n$-vertex graph has twin-width at most 4.

\begin{theorem}\label{thm:subd}
  Let $G$ be a trigraph obtained by subdividing each edge of an $n$-vertex graph~$H$ at least $2 \lceil \log n \rceil - 1$ times, and by turning red any subset of its edges as long as the red degree of $G$ remains at most~4, and no vertex with red degree~4 has a black neighbor.
  Then $\tww(G) \leqslant 4$.
\end{theorem}
\begin{proof}
  By no more than doubling the number of vertices of $H$, we can assume that $n$ is a power of~2.
  Indeed, padding $H$ with isolated vertices up to the next power of~2 does not change the quantity $\lceil \log |V(H)| \rceil$.

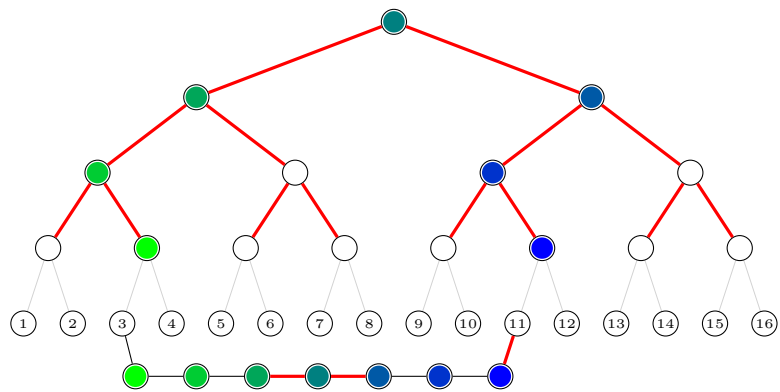
\begin{figure}[b]
    \centering
    \begin{tikzpicture}
      \def\s{1}     
      \def\sh{1.3}  
      \def\z{5}     
      \pgfmathtruncatemacro\zm{\z-1}
      \pgfmathtruncatemacro\zmm{\zm-1} 

    \foreach \i in {1,...,\z}{
      \pgfmathtruncatemacro\k{2^\i/2} 
      \foreach \j in {1,...,\k}{
        \pgfmathsetmacro\h{2^\z/(2 * 2^\i)}  
        \node[draw,circle] (a\i\j) at (\j * \h * \sh - \h * \sh / 2, - \i * \s) {} ;
      }
    }
    \pgfmathtruncatemacro\k{2^\z/2}
    \foreach \j in {1,...,\k}{
        \pgfmathsetmacro\h{1/2}  
        \node at (\j * \h * \sh - \h * \sh / 2, - \z * \s) {\tiny{$\j$}} ;
    }

    \foreach \i [count = \ip from 2] in {1,...,\zmm}{
      \pgfmathtruncatemacro\k{2^\i/2}
      \foreach \j in {1,...,\k}{
        \pgfmathtruncatemacro\jm{2 * \j - 1}
        \pgfmathtruncatemacro\jp{2 * \j}
        \draw[red,very thick] (a\i\j) -- (a\ip\jm) ;
        \draw[red,very thick] (a\i\j) -- (a\ip\jp) ;
        }
    }
    \pgfmathtruncatemacro\k{2^\zm/2}
      \foreach \j in {1,...,\k}{
        \pgfmathtruncatemacro\jm{2 * \j - 1}
        \pgfmathtruncatemacro\jp{2 * \j}
        \draw[very thin,opacity=0.18] (a\zm\j) -- (a\z\jm) ;
        \draw[very thin,opacity=0.18] (a\zm\j) -- (a\z\jp) ;
    }

    \foreach \i in {1,...,7}{
      \node[draw,circle] (b\i) at (1+0.8 * \i, -5.7) {} ;
    }
    \draw[red,very thick] (b7) -- (a\z11) ;
    \draw (a\z3) -- (b1) -- (b2) -- (b3) ;
    \draw[red, very thick] (b3) -- (b4) -- (b5) ;
    \draw (b5) -- (b6) -- (b7) ;

    \foreach \i/\op/\b in {4/1/50,3/0.7/35,5/0.7/65,2/0.4/20,6/0.4/80,1/0.2/0,7/0.2/100}{
      \node[circle,fill={blue!\b!green},inner sep=0.1cm] at (1+0.8 * \i, -5.7) {} ;
    }
    \foreach \i/\j/\op/\b in {1/1/1/50,2/1/0.7/35,3/1/0.4/20,4/2/0.2/0,2/2/0.7/65,3/3/0.4/80,4/6/0.2/100}{
      \node[circle,fill={blue!\b!green},inner sep=0.1cm] at (a\i\j) {} ;
    }
    \end{tikzpicture}
    \caption{Contracting the pairs of vertices with the same color, from the greenest to the bluest, is a partial 4-sequence, which acts as a deletion of the subdivided edge $(3,11)$.}
    \label{fig:subd}
  \end{figure}
  
  Let $G'$ be a supertrigraph of~$G$ obtained by arbitrarily arranging the vertices of $H$ (in~$G$) at the leaves of a ``virtual'' full binary tree of height $\log n$.
  So that the red degree does not exceed~4, we so far omit the edges of the tree incident to a leaf (i.e., a vertex of $H$), while we put in red all the other edges of the tree.
  (The missing edges of the tree will naturally appear in red.)
  The internal nodes of the tree are all fresh vertices, not present in $G$.
  See~\cref{fig:subd,fig:subd2} for an illustration.
  We show that $\tww(G') \leqslant 4$, hence by~\cref{obs:subtrigraph}, $\tww(G) \leqslant 4$ since $G$ is an induced subtrigraph of $G'$.

  We label $1, 2, \ldots, n$ the vertices of $H$.
  If there is an edge $ij \in E(H)$, it is subdivided into a path, say, $i, s(ij,1), s(ij,2), \ldots, s(ij,z), j$ in $G$ with $z \geqslant 2 \log n - 1$.
  First, we repeatedly contract adjacent vertices in the middle of this path until it consists of exactly $2 \log n$ edges.
  If $z > 2 \log n - 1$, we had to contract at least one pair of adjacent vertices.
  Thus the vertex in the middle of the path necessarily has now two red edges incident to it.
  Note that the other edges of the path can be black or red indifferently.
  To avoid cumbersome notations, we rename the inner vertices of the path $s(ij,1), s(ij,2), \ldots, s(ij,z)$ with now $z = 2 \log n - 1$.

  \begin{claim}\label{lem:hanging-garland}
    There is a partial 4-sequence from $G'$ to $G' - \{s(ij,1), \ldots, s(ij,z)\}$.
  \end{claim}
  \begin{proof}\renewcommand{\qedsymbol}{$\lhd$ }
    Intuitively we ``zip'' the subdivision of $ij$ with the walk made by the union of the path from leaf $i$ to the root, and the path from the root to leaf $j$.
    Let $i, v_1, v_2, \ldots, v_z, j$ be the concatenation of the simple path from $i$ to the root of the tree, and the simple path from the root to $j$.
    Its length is thus $2\log n = z+1$.
    For $h$ going from $1$ to $z = 2 \log n - 1$, we contract $v_h$ and $s(ij,h)$ (see~\cref{fig:subd}).
    After each contraction, the newly formed vertex has red degree at most 4.
    The red degree of vertices that are neither the new vertex nor a leaf of the tree is either unchanged or at most 2.
    The red degree of a leaf $\ell$ of the tree may increase by~1.
    This may only happen the first time a neighbor of $\ell$ is involved in a contraction, and that contraction merges a black neighbor of $\ell$ with the parent of $\ell$ in the tree (like is the case for leaf 3 from~\cref{fig:subd} to \cref{fig:subd2}).
    By assumption, this implies that $\ell$ had red degree at most~3, thus its red degree does not exceed~4.
    Thus, what we defined is indeed a partial 4-sequence.
    One can finally notice that after these $z$ contractions, we indeed reach trigraph $G' - \{s(ij,1), \ldots, s(ij,z)\}$.
\end{proof}

    \begin{figure}[h]
    \centering
    \begin{tikzpicture}
      \def\s{1}     
      \def\sh{1.3}  
      \def\z{5}     
      \pgfmathtruncatemacro\zm{\z-1}
      \pgfmathtruncatemacro\zmm{\zm-1}  

    \foreach \i in {1,...,\z}{
      \pgfmathtruncatemacro\k{2^\i/2}
      \foreach \j in {1,...,\k}{
        \pgfmathsetmacro\h{2^\z/(2 * 2^\i)}  
        \node[draw,circle] (a\i\j) at (\j * \h * \sh - \h * \sh / 2, - \i * \s) {} ;
      }
    }
    \pgfmathtruncatemacro\k{2^\z/2}
    \foreach \j in {1,...,\k}{
        \pgfmathsetmacro\h{1/2}  
        \node at (\j * \h * \sh - \h * \sh / 2, - \z * \s) {\tiny{$\j$}} ;
    }
    \node[draw,thick,rounded corners,fit=(a21)] {} ;

    \foreach \i [count = \ip from 2] in {1,...,\zmm}{
      \pgfmathtruncatemacro\k{2^\i/2}
      \foreach \j in {1,...,\k}{
        \pgfmathtruncatemacro\jm{2 * \j - 1}
        \pgfmathtruncatemacro\jp{2 * \j}
        \draw[red,very thick] (a\i\j) -- (a\ip\jm) ;
        \draw[red,very thick] (a\i\j) -- (a\ip\jp) ;
        }
    }
    \pgfmathtruncatemacro\k{2^\zm/2}
      \foreach \j in {1,...,\k}{
        \pgfmathtruncatemacro\jm{2 * \j - 1}
        \pgfmathtruncatemacro\jp{2 * \j}
        \draw[very thin,opacity=0.18] (a\zm\j) -- (a\z\jm) ;
        \draw[very thin,opacity=0.18] (a\zm\j) -- (a\z\jp) ;
    }
    
      \draw[red,very thick] (a\zm2) -- (a\z3) ;
      \draw[red,very thick] (a\zm2) -- (a\z3) ;  
    
    \foreach \i in {4,...,7}{
      \node[draw,circle] (b\i) at (1+0.8 * \i, -5.7) {} ;
    }
    \draw[red,very thick] (b7) -- (a\z11) ;
    \draw[red, very thick] (a21) to [bend left=8] (b4) ;
    \draw[red, very thick] (b4) -- (b5) ;
    \draw (b5) -- (b6) -- (b7) ;

    \foreach \i/\op/\b in {4/1/50,5/0.7/65,6/0.4/80,7/0.2/100}{
      \node[circle,fill={blue!\b!green},inner sep=0.1cm] at (1+0.8 * \i, -5.7) {} ;
    }
    \foreach \i/\j/\op/\b in {1/1/1/50,2/1/0.7/35,3/1/0.4/20,4/2/0.2/0,2/2/0.7/65,3/3/0.4/80,4/6/0.2/100}{
      \node[circle,fill={blue!\b!green},inner sep=0.1cm] at (a\i\j) {} ;
    }
    \end{tikzpicture}
    \caption{The picture after the first three contractions. The newly formed vertex has red degree~4.}
    \label{fig:subd2}
  \end{figure}

  We apply~\cref{lem:hanging-garland} for each edge of $H$ (or rather, subdivided edge in $G$).
  We are then left with a red full binary tree which admits a 3-sequence by~\cref{lem:tree}.
  Hence there is a 4-sequence for $G'$, and in particular, for $G$.
\end{proof}

We will only use the following consequence.
\begin{lemma}\label{lem:subd}
  Let $G$ be a trigraph obtained by subdividing at least $2 \lceil \log n \rceil - 1$ times each edge of an $n$-vertex graph $H$ of degree at most~4, and by turning red all its edges.
  Then $\tww(G) \leqslant 4$.
\end{lemma}

\section{Hardness of determining if the twin-width is at most four}\label{sec:hardness}

Here we show the main result of the paper.
\begin{reptheorem}{thm:main}
  Deciding if a graph has twin-width at most~4 is NP-complete.
  Furthermore, no algorithm running in time $2^{o(n/\log n)}$ can decide if an $n$-vertex graph has twin-width at most~4, unless the ETH fails.
\end{reptheorem}
The membership to NP is ensured by the $d$-sequence: a polynomial-sized certificate that a graph has twin-width at most~$d$, checkable in polynomial time.   
We thus focus on the hardness part of the statement, and design a quasilinear reduction from~\textsc{$3$-SAT}.

\subsection{Encoding a trigraph by a graph}\label{subsec:trigraph-encoding}

In this subsection, we present a construction allowing to encode trigraphs into (plain) graphs.
Our objective is, given a trigraph $H$ with red degree at most $d$, to produce a graph $G$ such that $H$ admits a $2d$-sequence iff $G$ admits a $2d$-sequence.

Formally, we build a graph $G$ for which every $2d$-sequence $\mathcal S$ inevitably creates $H$ as an induced subtrigraph of a trigraph of $\mathcal S$.
By~\cref{obs:subtrigraph}, the existence of a $2d$-sequence for $G$ implies the existence of one for $H$. Moreover, this construction is such that there is a partial $2d$-sequence which, from $G$, exactly reaches $H$.
Therefore, any $2d$-sequence for $H$ can be augmented to become a $2d$-sequence for $G$.


At first sight, our construction works only for trigraphs $H$ with a connected red graph.
We could show the following.
\begin{lemma}
  For every trigraph $H$ such that $\mathcal R(H)$ is connected and of degree at most~$d$, there is a graph $G$ such that:
  \begin{compactitem}
  \item every $2d$-sequence of $G$ goes through a supertrigraph of $H$, and
  \item there is a partial $2d$-sequence from $G$ to $H$.
  \end{compactitem}
  \label{lem:connected-encoding}
\end{lemma}

However, we will need a stronger version, allowing to simulate trigraphs with disconnected red graphs.
First we show how to turn the trigraph induced by one connected component of the red graph into a plain graph (this is~\cref{lem:trigraph-encoding}).
Later we loop through all connected components of the red graph with size at least~2, and apply~\cref{lem:trigraph-encoding} (that is done in~\cref{lem:reduction-trigraph}).


\begin{lemma}\label{lem:trigraph-encoding}
  Let $H$ be a trigraph, and $S \subseteq V(H)$ be the vertex set of a connected component of $\mathcal R(H)$ with degree at most~$d$.
  There is a trigraph $G$ and a set $T \subseteq V(G)$ satisfying the following statements:
  \begin{compactenum}
  \item $\card{T}$ is bounded by a function of $d$ and $\card{S}$ only,
  \item no red edge touches $T$ (in particular, $G[T]$ has no red edge),
  \item $G - T$ is isomorphic to $H - S$,
  \item there is a partial $2d$-sequence from $G$ to $H$, and 
  \item every $2d$-sequence of $G$ goes through a supertrigraph of some trigraph $\tilde{H}$, where $\tilde{H}$ can be obtained from $H$ by performing contractions \emph{not} involving vertices of $S$, nor creating red edges incident to $S$.
  \end{compactenum}
\end{lemma}
\begin{proof}
  We begin with the construction of $G$. The vertex set of $G$ can be split into two sets: on one hand $V(H)\setminus S$, and on the other hand $T$, which will be defined in function of $H[S]$.
  Let $\{v_1, v_2, \ldots, v_{\card{S}}\}$ be the set $S$.
  To build set $T$, we blow up every vertex $v_i$ of $S$ into a copy $L_i$ of the biclique $K_{t,t}$ with $t = 2(2d+2) \cdot (2d)^{\card{S}-1}+1$.
  The two \emph{sides} of each biclique $L_i$ are denoted by $A_i = \{a_{i,j}\}_{j \in [t]}$ and $B_i = \{b_{i,j}\}_{j \in [t]}$, and are respectively called \emph{$A$-side} and \emph{$B$-side}.
  The vertex set $T \subseteq V(G)$ is: $$T=\bigcup_{v_i \in S} V(L_i) = \bigcup_{v_i \in S, j \in [t]} \{a_{i,j}, b_{i,j}\}.$$

  In addition to the edges forming each biclique, we replace every black edge of $H[S]$ by a complete bipartite graph between the two corresponding bicliques. Furthermore, every red edge of $H[S]$ becomes a black canonical matching, with edges $a_{i,j}a_{i',j}$ and $b_{i,j}b_{i',j}$.
  Formally, $$E(G[T])=\bigcup_{v_i \in S} A_i \times B_i ~\cup \bigcup_{v_iv_{i'} \in E(H)} L_i \times L_{i'}~\cup~\bigcup_{\substack{v_iv_{i'} \in R(H)\\ j \in [t]}} \{a_{i,j}a_{i',j}, b_{i,j}b_{i',j}\}.$$
  
  We finish the construction by considering the black edges initially connecting $S$ and $V(H)\setminus S$ in the trigraph $H$.
  For any edge $v_iz \in E(H)$ with $v_i \in S$ and $z \in V(H)\setminus S$, we add a black edge between any vertex of the biclique $L_i$ and $z \in V(G)\setminus T$. In summary,
  $$E(G) = E(G[T]) \cup E(H - S) \cup ~\bigcup_{\substack{v_i \in S\\ z \in V(H)\setminus S\\ v_iz \in E(H)}} L_i \times \{z\}.$$
  
  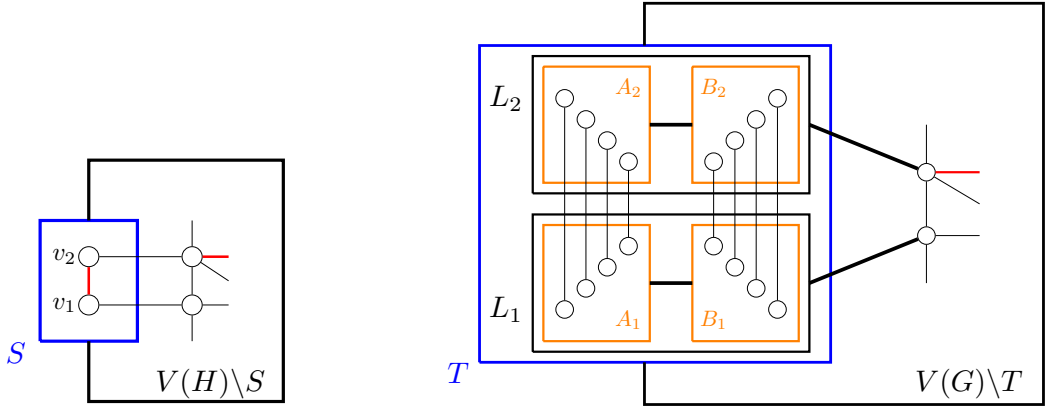
\begin{figure}[h]

\begin{subfigure}[b]{0.39\columnwidth}
\scalebox{0.8}{\begin{tikzpicture}
	\draw[line width = 1.5pt, color = blue] (1.0,2.0)--(2.6,2.0)--(2.6,4.0)--(1.0,4.0) -- (1.0,2.0);
	\draw[line width = 1.5pt] (1.8,2.0)--(1.8,1.0)--(5.0,1.0)--(5.0,5.0)--(1.8,5.0)--(1.8,4.0);
	\node[color = blue, scale = 1.4] at (0.6,1.8) {$S$};
	
	\node[draw,circle] (u) at (1.8,2.6) {};
	\node[draw,circle] (v) at (1.8,3.4) {};
        \node[scale = 1.2] (v1) at (1.4,2.6) {$v_1$};
	\node[scale = 1.2] (v2) at (1.4,3.4) {$v_2$};
	\draw[line width = 1.2pt, color = red] (u)--(v);
	
	\node[draw,circle] (x) at (3.5,2.6) {};
	\node[draw,circle] (y) at (3.5,3.4) {};
	\draw (u)--(x);
	\draw (v)--(y);
	\draw (x)--(y);
	\draw (x)--(3.5,2.0);
	\draw (x)--(4.1,2.6);
	\draw[line width = 1.2pt, color = red] (y)--(4.1,3.4);
	\draw (y)--(4.1,3.0);
	\draw (y)--(3.5,4.0);
	
	\node[scale = 1.4] at (3.8,1.3) {$V(H)\backslash S$};	
\end{tikzpicture}}
\caption{An example of trigraph $H$. Set $S$ consists of two vertices $v_1$ and $v_2$, linked by a red edge. Vertices of $H - S$ adjacent to $S$ are also represented.}
\end{subfigure}
~
\begin{subfigure}[b]{0.56\columnwidth}
\scalebox{0.7}{\begin{tikzpicture}
	\draw[line width = 1.5pt, color = blue] (0.4,1.0)--(7.0,1.0)--(7.0,7.0)--(0.4,7.0) -- (0.4,1.0);
	\draw[line width = 1.5pt] (3.5,7.0)--(3.5,7.8)--(11.0,7.8) -- (11.0,0.2) -- (3.5,0.2) -- (3.5,1.0);
	
	\draw[line width = 1.2pt] (1.4,1.2)--(6.6,1.2)--(6.6,3.8)--(1.4,3.8) -- (1.4,1.2);
	\draw[line width = 1.2pt] (1.4,4.2)--(6.6,4.2)--(6.6,6.8)--(1.4,6.8) -- (1.4,4.2);
	\draw[line width = 1.2pt, color = orange] (1.6,1.4)--(3.6,1.4)--(3.6,3.6)--(1.6,3.6) -- (1.6,1.4);
	\draw[line width = 1.2pt, color = orange] (1.6,6.6)--(3.6,6.6)--(3.6,4.4)--(1.6,4.4) -- (1.6,6.6);
	\draw[line width = 1.2pt, color = orange] (4.4,1.4)--(6.4,1.4)--(6.4,3.6)--(4.4,3.6) -- (4.4,1.4);
	\draw[line width = 1.2pt, color = orange] (4.4,6.6)--(6.4,6.6)--(6.4,4.4)--(4.4,4.4) -- (4.4,6.6);
	\draw[line width = 2pt] (3.6,5.5)--(4.4,5.5);
	\draw[line width = 2pt] (3.6,2.5)--(4.4,2.5);
	
	\foreach \i in {0,...,3}{
      \node[draw,circle] (a1\i) at (2.0+0.4*\i,6.0-0.4*\i) {};
      \node[draw,circle] (b1\i) at (6.0-0.4*\i,6.0-0.4*\i) {};
      \node[draw,circle] (a2\i) at (2.0+0.4*\i,2.0+0.4*\i) {};
      \node[draw,circle] (b2\i) at (6.0-0.4*\i,2.0+0.4*\i) {};
      \draw (a1\i)--(a2\i);
      \draw (b1\i)--(b2\i);
    }
	
	\node[draw,circle] (x) at (8.8,3.4) {};
	\node[draw,circle] (y) at (8.8,4.6) {};
	\draw (x) -- (y) ;
	\draw (x)--(8.8,2.5);
	\draw (x)--(9.8,3.4);
	\draw[line width = 1.2pt, color = red] (y)--(9.8,4.6);
	\draw (y)--(9.8,4.0);
	\draw (y)--(8.8,5.5);	
	
	\draw[line width = 2pt] (6.6,5.5)--(y);
	\draw[line width = 2pt] (6.6,2.5)--(x);

	\node[color = blue, scale = 1.6] at (0.0,0.8) {$T$};
	\node[scale = 1.6] at (0.9,2.0) {$L_1$};
	\node[scale = 1.6] at (0.9,6.0) {$L_2$};
	\node[scale = 1.2, color = orange] at (3.2,1.8) {$A_1$};
	\node[scale = 1.2, color = orange] at (4.8,1.8) {$B_1$};
	\node[scale = 1.2, color = orange] at (3.2,6.2) {$A_2$};
	\node[scale = 1.2, color = orange] at (4.8,6.2) {$B_2$};
	\node[scale = 1.6] at (9.6,0.6) {$V(G)\backslash T$};
\end{tikzpicture}}
\caption{Graph $G$ obtained from $H$. Thick black edges link every vertex of one endpoint to every vertex of the other endpoint. Set $T$ contains two bicliques $L_1$ and $L_2$, inherited from $S = \{v_1,v_2\}$.}
\end{subfigure}

\caption{An example of encoding: induced trigraph $S$ (left) and the plain graph $T$ (right).}
\label{fig:encoding}
\end{figure}
  
  \textbf{Statements 1-3.} The first three statements are satisfied by this construction.
  In particular the size of $T$ is $2t\card{S}$, and $t$ was defined as a function of $d$ and $\card{S}$.
  Moreover, no vertex of $V(H)\setminus S$ was modified by this construction, so $G - T$ is isomorphic to $H - S$.
  
  \textbf{Statement 4.} We focus now on the fourth statement: We exhibit a partial $2d$-sequence~$\mathcal{S}^*$ from $G$ to $H$.
  Let us begin with a short description of $\mathcal{S}^*$. This particular sequence never contracts two vertices of $T$ lying in different bicliques.
Moreover, after a contraction of two vertices belonging to the same biclique, it forces the same contraction in all other bicliques of $T$. We now describe $\mathcal{S}^*$ in detail.

We start by contracting $a_{i,1}$ and $a_{i,2}$ for each $v_i \in S$.
Then, we contract $\{a_{i,1},a_{i,2}\}$ with $a_{i,3}$ for each $v_i \in S$, and so on.
When the $A$-side of each biclique $L_i$ is contracted into a single vertex $\{a_{i,1},a_{i,2},\ldots,a_{i,t}\}$, we proceed similarly with the $B$-side.
Finally, each biclique~$L_i$ contains exactly two contracted vertices, which are respectively $A_i$ and $B_i$.
We contract all these pairs in each $L_i$: every biclique is now contracted into a single vertex.

The obtained trigraph is isomorphic to $H$.
Indeed, if $v_iv_{i'}$ is a black edge (resp. a~non-edge) in $H[S]$, then $L_i$ and $L_{i'}$ are homogeneous and, as contracted vertices, they are connected with a black edge (resp. a non-edge).
However, if $v_iv_{i'}$ is a red edge in $H[S]$, then there is a semi-induced matching connecting bicliques $L_i$ and $L_{i'}$ in $G[T]$, so the contracted vertices $L_i$ and $L_{i'}$ are linked by a red edge.
Finally, the black edges between $S$ and $V(H)\setminus S$ are preserved because the contractions of $\mathcal{S}^*$ occur inside bicliques of $T$ which are homogeneous (i.e., fully adjacent or fully non-adjacent) with every vertex $z \in V(H)\setminus S$.

Let us check that the partial sequence $\mathcal{S}^*$ only goes through trigraphs with red degree at most $2d$.
As we only contract vertices within the same side of each biclique first, no red edge appears between two vertices of the same biclique.
Moreover, if  $v_iv_{i'} \notin R(H)$, then bicliques $L_i$ and $L_{i'}$ are homogeneous, so no red edge appears between two vertices belonging respectively to $L_i$ and $L_{i'}$.
Thus, the red edges may only bridge two different bicliques $L_i, L_{i'}$ such that $v_iv_{i'} \in R(H)$.
The way we progressively contract the matching between $L_i$ and $L_{i'}$, a contracted vertex within $L_i$ is the endpoint of at most two red edges towards $L_{i'}$.
Assume without loss of generality that $u(G) = \{a_{i,1},a_{i,2},\ldots,a_{i,j+1}\}$ is contracted in $A_i$ while only $u'(G) = \{a_{i',1},a_{i',2},\ldots,a_{i',j}\}$ is contracted in $A_{i'}$ for some $j\in [t-1]$.
There are two red edges between bicliques $L_i$ and $L_{i'}$ which are $uu'$ and $uu''$, where $u''(G) = \{a_{i',j+1}\}$.
As the red degree of $H[S]$ is at most $d$, we indeed have that $\mathcal{S}^*$ is a partial $2d$-sequence.

\textbf{Statement 5.} We terminate by showing that any $2d$-sequence $\mathcal{S}$ from graph $G$ necessarily produces at some moment a supertrigraph of $\tilde{H}$, where $\tilde{H}$ is obtained from $H$ with contractions not involving set $S$.
As this part of the proof is more intricate than the previous ones, two intermediate claims are stated.
We fix $g(d,t) = (t-1)/(2d+2) = 2(2d)^{\card{S}-1}$.

\begin{claim}\label{clm:ctr-biclique}
  Let $G_h$ be any trigraph of $\mathcal S = G,\ldots,G_h,\ldots,K_1$.
  If a vertex $u$ of $G_h$ is such that $u(G)$ intersects two distinct bicliques $L_i$ and $L_{i'}$, then there is a vertex $w$ (possibly~$u$) of $G_h$ intersecting one side of $L_i$ in at least $g(d,t)$ elements: $|w(G) \cap A_i| \geqslant g(d,t)$ or $|w(G) \cap B_i| \geqslant g(d,t)$.
\end{claim}

\begin{proof}\renewcommand{\qedsymbol}{$\lhd$ }
  Let $a \in u(G) \cap V(L_i)$ and $z \in u(G) \cap V(L_{i'})$.
  We assume w.l.o.g. that the partite set of the biclique $L_i$ in which $a$ lies is $A_i$.
  We identify a set~$X \subset V(G)$ of size at least $t-1$ such that, for any vertex $u' \in V(G_h) \setminus \{u\}$ such that $u'(G) \cap X \neq \emptyset$, there is a red edge $uu'$.
  The composition of set $X$ depends on the adjacency between $v_i$ and $v_{i'}$ in $H$.
  \begin{compactitem}
  \item If $v_iv_{i'}$ is a non-edge in $H$, then we fix $X = B_i$.
    Indeed, any element of $B_i$ is adjacent to~$a$ but not to~$z$, so there is a red edge between $u$ and any contracted vertex $u'$ intersecting $B_i$.
    Moreover, $\card{B_i} = t$.
  \item If $v_iv_{i'}$ is a black edge in $H$, then we fix $X = A_i \setminus \{a\}$. Any element of $A_i \setminus \{a\}$ is not adjacent to $a$ but adjacent to $z$. Thus, any contracted vertex $u$ intersecting $A_i \setminus \{a\}$ is connected to $u$ in red. Moreover, $\card{A_i \setminus \{a\}} = t-1$.
  \item If $v_iv_{i'}$ is a red edge in $H$, then we fix $X = B_i \setminus \{z'\}$, where $z'$ is the matching neighbor of $z$ in $L_i$. Observe that $z'$ may belong to $A_i$ and, in this case, $X = B_i$.
    Any element of~$X$ is adjacent to $a$ but not to $z$. Moreover, $\card{X} \in \{t-1,t\}$.
  \end{compactitem}
  In summary, there is a set $X$ of at least $t-1$ vertices such that any contracted vertex $u' \in G_h$ intersecting $X$ is linked by a red edge to~$u$.
  Consequently, there cannot be more than $2d+1$ contracted vertices $w \in G_h$ intersecting $X$, otherwise the red degree of $u$ is necessarily larger than $2d$. For this reason, at least one vertex $w \in G_h$ which intersects $X$ has size at least $\lceil(t-1)/(2d+1)\rceil \ge g(d,t)$. In each case, either $X \subsetneq A_i$ or $X \subseteq B_i$, so either $\card{w(G) \cap A_i} \geqslant g(d,t)$ or $\card{w(G) \cap B_i} \geqslant g(d,t)$.
\end{proof}

Let $G'$ be the first trigraph of sequence $\mathcal S$ having a vertex $u$ such that $|u(G) \cap T| \geqslant g(d,t)$. We denote by $G''$ the trigraph preceding $G'$ in $\mathcal S$.
We show that every contracted vertex $w \in G'$ is either fully contained in some biclique $L_i$ or contained in $V(G)\setminus T$.

\begin{claim}\label{clm:ctr-inside-T}
  Every contracted vertex $u \in V(G')$ satisfies:
  \begin{compactitem}
  \item either $u(G) \subseteq A_i$ or $u(G) \subseteq B_i$ for some $L_i$ of $T$,
  \item or $u(G) \subseteq V(G)\setminus T$.
  \end{compactitem}
  Furthermore, there is no red edge between a~contracted vertex inside a biclique of $T$ and a~contracted vertex in $G - T$.
\end{claim}
\begin{proof}\renewcommand{\qedsymbol}{$\lhd$ }
We begin by proving that if a~contracted vertex $w \in V(G')$ intersects $T$, then it intersects at most one of its bicliques, that is, $w(G) \cap T \subseteq L_i$ for some index $i$.
According to~\cref{clm:ctr-biclique}, all the contracted vertices of $G''$ intersect $T$ within a single biclique.
This holds in particular for the two vertices, say $x,y \in V(G'')$, whose contraction yields~$u$.
Assume that $x$ and $y$ are both contained within different bicliques, say, $L_i$ and $L_{i'}$, respectively.
By~\cref{clm:ctr-biclique}, there is a vertex of $G'$ intersecting $L_i$ in at least $g(d,t)$ elements. By definition of trigraph $G'$, this vertex is necessarily $u$. This yields a contradiction, as the cardinality of $x(G) = u(G) \cap V(L_i)$ would be at least $g(d,t)$. Therefore, vertex $u$ intersects at most one biclique.
In brief, from now on, every set $w(G)$ with $w \in V(G')$ can be written as the union of elements of some biclique $L_i$ with vertices of $G - T$.

Each side, $A_i$ or $B_i$, of a biclique $L_i$ satisfies the following property: There are at least $2d+2$ contracted vertices of $G'$ intersecting every $A_i$ (or $B_i$).
Any contracted vertex in $G''$ covers less than $g(d,t)$ elements of $T$, by definition.
So, at least $2d+3$ contracted vertices of $G''$ must intersect $A_i$ as $\card{A_i} = t$.
In $G'$, this property holds except for vertex $u$, which is the contraction of two vertices $x,y \in V(G'')$.
Therefore, as announced, at least $2d+2$ contracted vertices $w$ of $G'$ satisfy $w(G) \cap A_i \neq \emptyset$ (resp.~$w(G) \cap B_i \neq \emptyset$) for every side $A_i$ (resp.~$B_i$).

Based on this property, we show that, for any $w \in V(G')$, set $w(G)$ cannot contain both vertices of $T$ and of $V(G)\setminus T$.
Assume it is the case and w.l.o.g. that $w(G) \cap A_i \neq \emptyset$. All vertices of $L_i$ play the same role regarding the adjacencies between $T$ and $G - T$: either they are all connected to some $z \in G - T$ (if $v_iz \in E(H)$) or none of them is connected to $z$ (if $v_iz \notin E(H)$). Consequently, we can distinguish only two cases. First, assume that at least one vertex of $w(G)\setminus A_i$ is adjacent (black or red) to $L_i$ in $G$ (Figure~\ref{subfig:inside_T_case1}).
In that case, $w$~admits at least $2d+1$ red neighbors in $G'$ which are the other contracted vertices intersecting $A_i$.
Indeed, these vertices contain elements of $A_i$ which are adjacent to $w(G)\setminus A_i$ but not to $w(G) \cap A_i$. Second, assume that no edge of $G$ connects $w(G)\setminus A_i$ with $L_i$ (Figure~\ref{subfig:inside_T_case2}). We can identify $2d+1$ red neighbors of $w$: among the at least $2d+2$ contracted vertices of $G'$ intersecting $B_i$, at least $2d+1$ of them are different from $w$. These contracted vertices contain elements of $B_i$ which are adjacent to $w(G) \cap A_i$ but not to $w(G)\setminus A_i$. In summary, for any $w \in G'$, either $w(G) \subseteq L_i$ or $w(G) \subseteq V(G)\setminus T$.  

\begin{figure}[h]
\begin{subfigure}[b]{0.49\columnwidth}
\scalebox{0.6}{\begin{tikzpicture}
	\draw[line width = 1.5pt, color = blue] (8.5,1.0)--(14.5,1.0)--(14.5,7.0)--(8.5,7.0) -- (8.5,1.0);
	\draw[line width = 1.5pt] (11.5,7.0)--(11.5,7.8)--(4.0,7.8) -- (4.0,0.2) -- (11.5,0.2) -- (11.5,1.0);
	
	\draw[line width = 1.2pt] (9.5,1.2)--(12.1,1.2)--(12.1,6.4)--(9.5,6.4) -- (9.5,1.2);
	\draw[line width = 1.2pt, color = orange] (9.7,1.4)--(11.9,1.4)--(11.9,3.4)--(9.7,3.4) -- (9.7,1.4);
	\draw[line width = 1.2pt, color = orange] (9.7,4.2)--(11.9,4.2)--(11.9,6.2)--(9.7,6.2) -- (9.7,4.2);
	\draw[line width = 2pt] (10.8,3.4)--(10.8,4.2);
	
	\foreach \i in {0,...,3}{
      \node[draw,circle] (a2\i) at (10.2+0.4*\i,1.8+0.4*\i) {};
      \node[draw,circle] (b2\i) at (10.2+0.4*\i,5.8-0.4*\i) {};
    }
	
	
	
	\draw[color = green] (10.4,1.65)--(10.4,1.95)--(5.0,1.95)--(5.0,0.7) -- (8.0,0.7) -- (8.0,1.65) -- (10.4,1.65);
	\draw[color = green] (10.8,2.05)--(10.8,2.35)--(5.5,2.35)--(5.5,2.05) -- (10.8,2.05);
	\draw[color = black!40!green, line width = 1.4pt] (11.6,2.45)--(11.6,3.25)--(6.2,3.25)--(6.2,2.45) -- (11.6,2.45);
	
	\draw[color = green] (11.7,4.45)--(11.7,4.75)--(11.1,4.75)--(11.1,4.45) -- (11.7,4.45);
	\draw[color = green] (11.2,4.85)--(11.2,5.55)--(5.8,5.55)--(5.8,4.2) -- (6.8,4.2) -- (6.8,4.85) -- (11.2,4.85);
	\draw[color = green] (10.4,5.65)--(10.4,6.05)--(5.0,6.05)--(5.0,5.65) -- (10.4,5.65);
	\node[color = black!40!green, scale = 1.4] at (7.25,3.5) {$w(G)$};
	\node[draw,circle] (w) at (7.0,2.75) {};
	\draw[line width = 2pt] (9.5,3.7)--(w);
	
	\draw[line width = 2pt, color = red] (6.2,2.8)--(5.7,2.35);
	\draw[line width = 2pt, color = red,rounded corners] (6.2,2.9)--(5.2,2.9)--(5.2,1.95);
	
	\node[color = blue, scale = 1.6] at (14.9,0.8) {$T$};
	\node[scale = 1.6] at (12.5,3.8) {$L_i$};
	\node[scale = 1.2, color = orange] at (11.4,1.8) {$A_i$};
	\node[scale = 1.2, color = orange] at (11.4,5.8) {$B_i$};
	\node[scale = 1.6] at (5.4,7.4) {$V(G)\backslash T$};
\end{tikzpicture}}
\caption{Case when $v_i$ is adjacent to some vertex of $w(G)\setminus A_i$ in $H$.}
\label{subfig:inside_T_case1}
\end{subfigure}
~
\begin{subfigure}[b]{0.49\columnwidth}
\scalebox{0.6}{\begin{tikzpicture}
	\draw[line width = 1.5pt, color = blue] (8.5,1.0)--(14.5,1.0)--(14.5,7.0)--(8.5,7.0) -- (8.5,1.0);
	\draw[line width = 1.5pt] (11.5,7.0)--(11.5,7.8)--(4.0,7.8) -- (4.0,0.2) -- (11.5,0.2) -- (11.5,1.0);
	
	\draw[line width = 1.2pt] (9.5,1.2)--(12.1,1.2)--(12.1,6.4)--(9.5,6.4) -- (9.5,1.2);
	\draw[line width = 1.2pt, color = orange] (9.7,1.4)--(11.9,1.4)--(11.9,3.4)--(9.7,3.4) -- (9.7,1.4);
	\draw[line width = 1.2pt, color = orange] (9.7,4.2)--(11.9,4.2)--(11.9,6.2)--(9.7,6.2) -- (9.7,4.2);
	\draw[line width = 2pt] (10.8,3.4)--(10.8,4.2);
	
	\foreach \i in {0,...,3}{
      \node[draw,circle] (a2\i) at (10.2+0.4*\i,1.8+0.4*\i) {};
      \node[draw,circle] (b2\i) at (10.2+0.4*\i,5.8-0.4*\i) {};
    }
	
	
	
	\draw[color = green] (10.4,1.65)--(10.4,1.95)--(5.0,1.95)--(5.0,0.7) -- (8.0,0.7) -- (8.0,1.65) -- (10.4,1.65);
	\draw[color = green] (10.8,2.05)--(10.8,2.35)--(5.5,2.35)--(5.5,2.05) -- (10.8,2.05);
	\draw[color = black!40!green, line width = 1.4pt] (11.6,2.45)--(11.6,3.25)--(6.2,3.25)--(6.2,2.45) -- (11.6,2.45);
	
	\draw[color = green] (11.7,4.45)--(11.7,4.75)--(11.1,4.75)--(11.1,4.45) -- (11.7,4.45);
	\draw[color = green] (11.2,4.85)--(11.2,5.55)--(5.8,5.55)--(5.8,4.2) -- (6.8,4.2) -- (6.8,4.85) -- (11.2,4.85);
	\draw[color = green] (10.4,5.65)--(10.4,6.05)--(5.0,6.05)--(5.0,5.65) -- (10.4,5.65);
	\node[color = black!40!green, scale = 1.4] at (7.25,3.5) {$w(G)$};
	
	\draw[line width = 2pt, color = red, rounded corners] (9.0,3.25)--(9.0,4.6) -- (11.1,4.6);
	\draw[line width = 2pt, color = red] (6.5,3.25)--(6.5,4.2);
	\draw[line width = 2pt, color = red, rounded corners] (6.2,2.8)--(5.2,2.8)--(5.2,5.65);
	
	\node[color = blue, scale = 1.6] at (14.9,0.8) {$T$};
	\node[scale = 1.6] at (12.5,3.8) {$L_i$};
	\node[scale = 1.2, color = orange] at (11.4,1.8) {$A_i$};
	\node[scale = 1.2, color = orange] at (11.4,5.8) {$B_i$};
	\node[scale = 1.6] at (5.4,7.4) {$V(G)\backslash T$};
\end{tikzpicture}}
\caption{Case when $v_i$ is not adjacent to any vertex of $w(G)\setminus A_i$ in $H$}
\label{subfig:inside_T_case2}
\end{subfigure}

\caption{How contractions intersecting both $T$ and $V(G)\setminus T$ would imply large red degree in $G'$. The green boxes represent contracted vertices in $G'$.}
\label{fig:contractions_inside_T}
\end{figure}
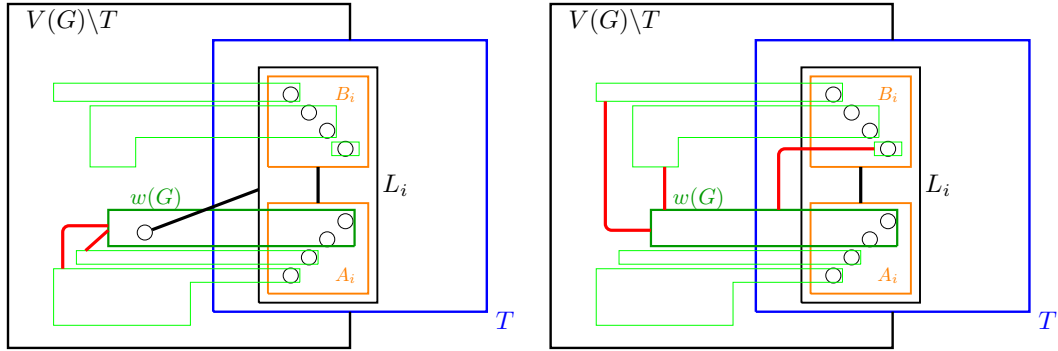

We now prove that there is no red edge $ww' \in R(G')$ such that $w(G) \subseteq L_i$ and $w'(G) \subseteq V(G)\setminus  T$. Suppose by way of contradiction that such an edge exists. As all vertices of $L_i$ play the same role in the adjacencies between $T$ and $G - T$, there is a red edge between $w'$ and any vertex of $G'$ intersecting $L_i$. As at least $2d+2$ of these vertices intersect $A_i$, the red degree of $w'$ is at least $2d+2$, a contradiction.

To end the proof, we show that any contracted vertex $w(G) \subseteq L_i$ verifies either $w(G) \subseteq A_i$ or $w(G) \subseteq B_i$.
Consider the at least $2d+1$ vertices of $G'$ intersecting $A_i$ different from $w$.
They contain elements of $A_i$ which are adjacent to $w(G) \cap A_i$ but not to $w(G) \cap B_i$.
Therefore, the red degree of $w$ is at least $2d+1$, a contradiction.
\end{proof}

In brief, the partial sequence from $G$ to $G'$ does not use contractions mixing vertices of~$T$ and vertices of $V(G)\setminus T$.
As a consequence, $G'$ can be split in two parts: its induced subtrigraph $G_T'$ obtained from contractions on the bicliques and its induced subtrigraph $G_{G - T}'$ obtained from the contractions on $G - T$.
Our objective is to prove that $H$ appears in the first part.

We now focus more specifically on the vertex $u$.
W.l.o.g., we assume that $u(G)$ intersects $A_i \subset L_i$.
We denote by $J_0$ a set of exactly $g(d,t)$ distinct indices $j$ such that $a_{i,j} \in u(G)$.
For every $J \subseteq [t]$, we denote by $A_{i,J}$ the set $\{a_{i,j} : j \in J\}$.
It holds that $|J_0|=|A_{i,J_0}|=g(d,t)$.

\begin{claim}\label{clm:H-appears}
  $H[S]$ is an induced subtrigraph of $G_T'$.
\end{claim}

\begin{proof}\renewcommand{\qedsymbol}{$\lhd$ }
  We initialize a subset $Y \subseteq S$ to $\{v_i\}$, an injective mapping $\rho: Y \hookrightarrow V(G')$ to $v_i \mapsto u$, and a subset $J \subseteq [t]$ to $J_0$.
  We maintain the following invariant:
  \[(\ast)~A_{h,J} \subseteq \rho(v_h)(G)~\text{for every}~v_h \in Y,~\text{and}~2 \leqslant |J| = \frac{t-1}{(2d+2) \cdot (2d)^{|Y|-1}}.\]
  
The invariant initially holds by construction of $J_0$.
While $Y \neq S$, we pick a pair $v_h \in Y, v_{h'} \in S \setminus Y$ such that $v_hv_{h'} \in R(H)$.
Since $v_h(G) \supseteq A_{h,J}$ and there is an induced matching between $A_{h,J}$ and $A_{h',J}$, the set $A_{h',J}$ can be spanned by at most $2d$ parts of $\mathcal P(G')$.
Thus we select of vertex $w_{h'} \in V(G')$ and a subset $J' \subseteq J$ of size $\frac{|J|}{2d}$ such that $w_{h'}(G) \supseteq A_{h',J'}$.
We set the new $J$ to $J'$, and we augment $\rho$ with $v_{h'} \mapsto w_{h'}$.
Finally, we add $v_{h'}$ to $Y$.

The invariant $(\ast)$ is preserved by construction.
Since the graph $\mathcal R(H)$ is connected, this process ends when $Y=S$ and $|J|=\frac{g(d,t)}{(2d)^{|S|-1}}=2$.
We claim that trigraph~$G'[\rho(S)]$, which is an induced subtrigraph of $G_T'$, is isomorphic to~$H[S]$.
Indeed, let $v_h, v_{h'} \in S$ be any pair of vertices, and let $w_h=\rho(v_h)$ and $w_{h'}=\rho(v_{h'})$.
Assume first that $v_hv_{h'} \in R(H)$. 
As $w_h(G) \supseteq A_{h,J}$, $w_{h'}(G) \supseteq A_{h',J}$, $|J| = 2$, and there is an induced matching between $A_{h,J}$ and $A_{h',J}$, there is a red edge in $G'$ between $w_h$ and $w_{h'}$.
If $v_hv_{h'}$ is a black edge in $H[S]$, then any pair of vertices in $w_h(G) \times w_{h'}(G)$ is a black edge in $G$. Hence, sets $w_h(G)$ and $w_{h'}(G)$ are homogeneous and $w_hw_{h'}$ is a black edge in $G'$. Eventually, if $v_hv_{h'}$ is a non-edge in $H[S]$, both $w_h(G)$ and $w_{h'}(G)$ are homogeneous in the sense that they are not adjacent at all. So, $w_hw_{h'}$ is a non-edge in $G'$. In brief, $G'[\rho(S)]$ is isomorphic to $H[S]$.
\end{proof}

This concludes the proof of~\cref{lem:trigraph-encoding}.
Indeed, consider the subtrigraph of $G'$ induced on vertices of both sets $\rho(S)$  and $G_{G - T}'$.
The subtrigraph $G'[\rho(S)]$ is isomorphic to $H[S]$ (\cref{clm:H-appears}) and the second part $G_{G - T}'$ is obtained from contractions on $G - T$ which do not make red edges appear towards $S$ (\cref{clm:ctr-inside-T}).
\end{proof}

If the red graph $\mathcal{R}(H)$ is not connected, one can apply~\cref{lem:trigraph-encoding} for each of its connected components.

\begin{lemma}\label{lem:reduction-trigraph}
  Given any trigraph $H$ whose red graph has degree at most~$d$ and connected components of size at most~$h$, one can compute in time $O_{d,h}(|V(H)|)$ a graph $G$ on $O_{d,h}(|V(H)|)$ vertices such that $H$ has a $2d$-sequence if and only if $G$ has a $2d$-sequence. 
\end{lemma}
\begin{proof}
  Let $S_1,\ldots,S_r$ be the connected components of $\mathcal{R}(H)$ of size at least~2.
  We define a collection of trigraphs $H_0=H,H_1,\ldots,H_r$, where $H_r$ is a plain graph isomorphic to $G$.
  Each trigraph $H_{\ell+1}$ will be built from $H_{\ell}$ by applying~\cref{lem:trigraph-encoding} on $S_{\ell+1}$.
  We proceed by induction and prove, for any $0\le \ell \le r$, that:
\begin{compactitem}
\item $\card{V(H_{\ell})}$ is linear in $\card{V(H)}$ if $d,h = O(1)$: $\card{V(H_{\ell})} = f_{\ell}(d,h)\card{V(H)}$,
\item the connected components of $\mathcal R(H_{\ell})$ of size at least~2 have vertex sets $S_{\ell+1}, \ldots, S_r$,
\item $H_{\ell}$ has a $2d$-sequence iff $H$ has a $2d$-sequence.
\end{compactitem}

The base case is trivial, as $H_0 = H$.
Let us assume, for the induction step, that $H_{\ell}$ satisfies these three properties.
We consider the red component $S_{\ell+1}$ in $H_{\ell}$ to build trigraph $H_{\ell+1}$.
We apply~\cref{lem:trigraph-encoding} with trigraph $H_{\ell}$ and vertex set $S_{\ell+1}$. 
The trigraph $H_{\ell+1}$ substitutes $H_{\ell}[S_{\ell+1}]$ with a plain graph $H_{\ell+1}[T]$ not touched by a red edge.
Thus the vertex sets of connected components of $\mathcal R(H_{\ell+1})$ of size at least~2 are those of $\mathcal R(H_\ell)$ minus $S_{\ell+1}$.

Moreover, the size of $T$ is bounded by $d$ and $|S_{\ell+1}|$ only.
As a consequence, 
 $$\card{V(H_{\ell+1})} = \card{T} + \card{V(H_{\ell})\setminus S_{\ell+1}} \le h(d,\card{S_{\ell+1}}) + f_{\ell}(d,h)\card{V(H)} \le f_{\ell+1}(d,h)\card{V(H)},$$
where $h$ is the function defined accordingly to~\cref{lem:trigraph-encoding}.

We finally show that $H_{\ell+1}$ has a $2d$-sequence iff $H_{\ell}$ has a $2d$-sequence.
This is a direct consequence of~\cref{lem:trigraph-encoding}.
On one hand, there is a $2d$-partial sequence from $H_{\ell+1}$ to $H_{\ell}$, so if $H_{\ell}$ has a $2d$-sequence, then $H_{\ell+1}$ has one, too.
On the other hand, if $H_{\ell+1}$ has a $2d$-sequence, then one trigraph of the sequence is a supertrigraph of $\tilde{H}_{\ell}$, where $\tilde{H}_{\ell}$ is a trigraph which can be obtained from $H_{\ell}$ by performing contractions preserving $S_{\ell}$ as a red connected component of the trigraph.
One can deduce from this statement a $2d$-sequence of $H_{\ell}$: first contract $H_{\ell}$ to obtain $\tilde{H}_{\ell}$, then conclude by~\cref{obs:subtrigraph}.
 
As a conclusion, trigraph $G=H_r$ satisfies the induction hypotheses: its red graph is edgeless and it has a $2d$-sequence if and only if $H$ admits one, too.
\end{proof}

In the sequel, we will use~\cref{lem:reduction-trigraph} with $d=2$ and $h=12$.
In particular, the size of the encoding graph will be linear in the trigraph.
As paths are connected graphs of degree at most~2, we will invoke this scaled-down version.

\begin{lemma}\label{lem:reduction-trigraph-sd}
  Given any trigraph $H$ whose red graph is a disjoint union of 12-vertex paths and isolated vertices, one can compute in polynomial time a graph $G$ on $O(|V(H)|)$ vertices such that $H$ has twin-width at most~4 if and only if $G$ has twin-width at most~4. 
\end{lemma}

\subsection{Foreword to the reduction}\label{subsec:foreword}

Our task is now slightly simpler.
Given a \textsc{$3$-SAT} instance $I$, we may design a \emph{trigraph} satisfying the requirements of~\cref{lem:reduction-trigraph-sd} with twin-width at most 4 if and only if $I$ is satisfiable.

As already mentioned, given an instance $I$ of \textsc{$3$-SAT} we will create an equivalent instance of the problem: \emph{is trigraph $G$ of twin-width at most~4?}  
We now present the various gadgets used in our reduction.
The only building block featuring red edges is the fence gadget.
Its red graph is a 12-vertex path that is a connected component in the red graph of the overall construction $G$.
Hence $G$ can be turned into a graph with only a constant multiplicative blow-up, by~\cref{lem:reduction-trigraph-sd}.

The correctness of the reduction naturally splits into two implications:
\begin{compactitem}
\item($i$) If $G$ admits a 4-sequence, then $I$ is satisfiable. 
\item($ii$) If $I$ is satisfiable, then $G$ admits a 4-sequence.
\end{compactitem}
We will motivate all the gadgets along the way, by exhibiting key properties that they impose on a potential 4-sequence.
These properties readily leads to a satisfying assignment for $I$.
So the proof of $(i)$ will mostly consist of aggregating lemmas specific to individual gadgets.

We also describe partial 4-sequences to reduce most of the gadgets.
However some preconditions (specifying the context in which a particular gadget stands) tend to be technical, and make more sense after the construction of $G$.
In those cases, to avoid unnecessarily lengthy lemmas, we only give an informal strategy, and postpone the adequate contraction sequence to the final proof of $(ii)$.   

\subsection{Fence gadget} \label{subsection:fg}

We now design a gadget $F$ partitioned into two vertex sets $A, B$.
The gadget is \emph{attached} to a non-empty subset $S \subseteq V(G) \setminus (A \cup B)$ by making $A$ and $S$ fully adjacent, and $B$ and $S$ fully non-adjacent.
Our intent is that, in a 4-sequence, a vertex of $F$ can be contracted with another vertex of the graph only when $S$ has been contracted into a single vertex.
We will ensure that every vertex outside $S \cup V(F)$ is fully adjacent to $B$ or fully non-adjacent to $A$.
As $|A|=|B|=6>4$ will hold, the effect is that vertices of $S$ have to be contracted to a single vertex before any vertex of $S$ (or subset of vertices within $S$) can be contracted with a vertex or subset of vertices in $G - S$.
To summarize, $F$ encloses $S$ into an ``unbreakable unit:'' the inside of $S$ cannot be contracted with the outside as long as $S$ is not a single vertex.  
Hence we call $F$ a \emph{fence gadget}.

The fence gadget is defined as a trigraph whose red graph is a simple path on 12 vertices, in particular, a connected graph of maximum degree~2.
In the overall trigraph $G$, no red edge will link a vertex in the gadget to a vertex outside of it.
Thus the fence gadget can be replaced by a graph by~\cref{lem:reduction-trigraph-sd}.

We now define the fence gadget.
Its vertex set is $A \cup B$ with $A=\{a_1,a_2,a_3,a_4,a_5,a_6\}$ and $B=\{b_1,b_2,b_3,b_4,b_5,b_6\}$.
Its black edge set consists of 13 edges: the cycles $a_1a_2a_3a_4a_5a_6a_1$ and $b_1b_2b_3b_4b_5b_6b_1$, plus the edge $b_1a_6$.
Its red edge set consists of 11 edges: $a_ib_i$ for each $i \in [6]$, and $a_ib_{i+1}$ for each $i \in [5]$.
Finally $A$ is made fully adjacent to $S$.
See~\cref{fig:fence-gadget} for an illustration.

\begin{figure}[ht!]
  \centering
  \begin{tikzpicture}
    \foreach \i in {1,...,6}{
      \node[draw,circle,inner sep=0.03cm] (b\i) at (360 * \i/6:1) {$b_\i$} ;
      \node[draw,circle,inner sep=0.03cm] (a\i) at (360 * \i/6:2) {$a_\i$} ;
      \draw[very thick,red] (a\i) -- (b\i) ;
    }
    \foreach \i [count = \ip from 2] in {1,...,5}{
      \draw (a\i) -- (a\ip) ;
      \draw (b\i) -- (b\ip) ;
      \draw[very thick,red] (a\i) -- (b\ip) ;
    }
    \draw (a1) -- (a6) ;
    \draw (b1) -- (b6) ;
    \draw[thick] (a6) to node[above]{$e$} (b1) ;
    
  \end{tikzpicture}
  \caption{The fence gadget $F$, with $A=\{a_i~|~1 \leqslant i \leqslant 6\}$ and $B=\{B_i~|~1 \leqslant i \leqslant 6\}$.}
  \label{fig:fence-gadget}
\end{figure}
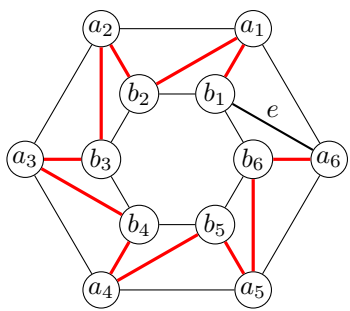

We will later nest fence gadgets.
Thus we have to tolerate that $F$ has other neighbors than $S$ in $G$.
Actually we even allow $V(F)$ to have neighbors outside of $S$ and the fence gadgets surrounding $F$.
We however always observe the following rule.

\begin{definition}[Attachment rule]\label{def:attachment-rule}
  A fence gadget $F$ with vertex bipartition $(A,B)$, and attached to $S$, satisfies the \emph{attachment rule} in a trigraph $H$ if $F$ is a connected component of~$\mathcal R(H)$, and there is a set $X \subseteq V(H) \setminus (A \cup B \cup S)$ such that:
  \begin{compactitem}
  \item $\forall x \in A$, $N(x) \setminus V(F) = X \cup S$,
  \item $\forall x \in B$, $N(x) \setminus V(F) = X$, and
  \item $\forall x \in X$, $S \subset N(x)$.
  \end{compactitem}
\end{definition}
Initially in $G$, we make sure that all the fence gadgets satisfy the attachment rule.
This will remain so until we decide to contract them.

We denote by $Y$ the set $V(G) \setminus (V(F) \cup S \cup X)$, when $X$ is defined in the attachment rule.
We make the following observations on the three possible neighborhoods that vertices outside of $F$ have within $V(F)$. 

\begin{observation}\label{obs:fg5} The fence gadget definition and the attachment rule implies:
  \begin{compactitem}
  \item $\forall x \in S$, it holds $N(x) \cap V(F) = A$,
  \item $\forall x \in X$, it holds $N(x) \cap V(F) = V(F) = A \cup B$, and
  \item $\forall x \in Y$, it holds $N(x) \cap V(F) = \emptyset$.
  \end{compactitem}
\end{observation}

Henceforth we will represent every fence gadget as a brown rectangle surrounding the set~$S$ it is attached to.
The vertices of $X$ are linked to the brown rectangle, as they are fully adjacent to $S \cup V(F)$.
See~\cref{fig:fence-ar} for an illustration of the attachment rule, and a compact representation of fence gadgets, and~\cref{fig:fence-nesting} for two nested fence gadgets and how the attachment rule is satisfied for both. 
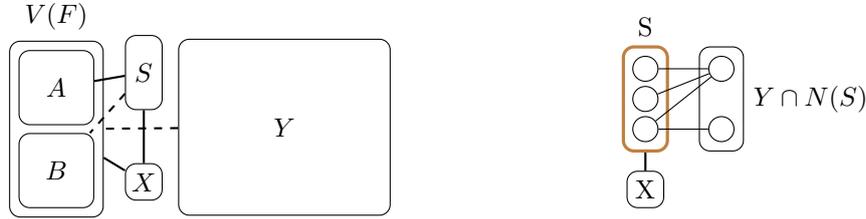
\begin{figure}[h!]
  \centering
  \begin{tikzpicture}
    \foreach \l/\i/\j in {a1/0/0,a2/0.5/0.5,b1/0/-0.6,b2/0.5/-1.1, s1/1.4/0.2,s2/1.4/0.7, x/1.4/-1, y1/2.1/-1.2,y2/4.4/0.65}{
      \node (\l) at (\i,\j) {} ;
    }
    \foreach \i/\j/\w in {{(a1) (a2)}/A/A,{(b1) (b2)}/B/B,{(A) (B)}/F/, {(s1) (s2)}/S/S, (x)/X/X, {(y1) (y2)}/Y/Y}{
      \node[draw,rounded corners,fit=\i,label=center:$\w$] (\j) {} ;
    }
    \node at (0.25,1.2) {$V(F)$} ;

    \foreach \i/\j/\k in {S/A/, S/X/, X/F/, Y/F/dashed, S/B/dashed}{
      \draw[thick,\k] (\i) -- (\j) ;
    }

    \begin{scope}[xshift=8cm, yshift=-0.5cm]
    \foreach \l/\i/\j/\o in {s1/0/0.2,s2/0/0.6,s3/0/1, y1/1/0.2,y2/1/1}{
      \node[draw,circle] (\l) at (\i,\j) {} ;
    }
    \node (x) at (0,-0.6) {} ;
    \node[draw,rounded corners,fit=(x),label=center:X] (X) {} ;
    \node[draw,rounded corners,fit=(y1) (y2),label=right:$Y \cap N(S)$] (Y) {} ;
    \node[draw,very thick,brown,rounded corners,fit=(s1) (s2) (s3),label=S] (S) {} ;
    \draw[thick] (X) -- (S) ;
    \foreach \i/\j in {s1/y2,s2/y2,s3/y2,s1/y1}{
      \draw (\i) -- (\j) ;
    }
    \end{scope}
  \end{tikzpicture}
  \caption{Left: The forced adjacencies (solid lines, all edges between the two sets) and non-adjacencies (dashed lines, no edge between the two sets), as specified by the attachment rule.
  Right: Symbolic representation of the fence gadget attached to $S$ by a brown rectangle; the vertices in $X$ are linked to the brown box, while the potential neighbors of $S$ in $Y$ are only linked individually to their neighbors in $S$, and are fully non-adjacent to the vertices of the fence $V(F)$. Possible edges between $X$ and $Y$ are \emph{not} represented.}
  \label{fig:fence-ar}
  
\end{figure}  
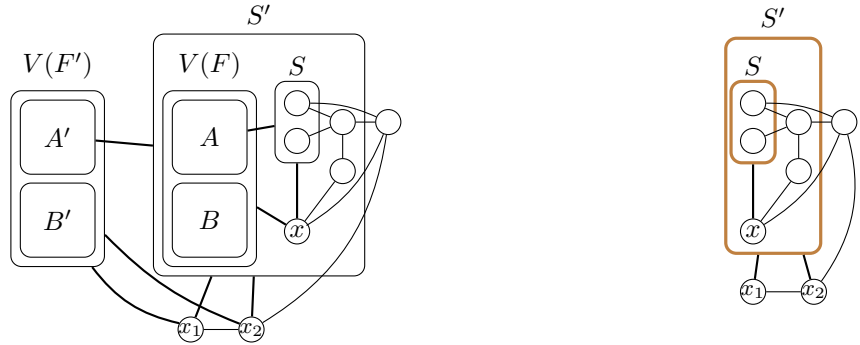
\begin{figure}[h!]
  \centering
  \begin{tikzpicture}
    \foreach \l/\i/\j in {a1/0/0,a2/0.5/0.5,b1/0/-0.6,b2/0.5/-1.1, ap1/-2/0,ap2/-1.5/0.5,bp1/-2/-0.6,bp2/-1.5/-1.1, x/1.4/-1, y1/2.1/-1.2,y2/4.4/0.65}{
      \node (\l) at (\i,\j) {} ;
    }
    \foreach \l/\i/\j in {s1/1.4/0.2,s2/1.4/0.7, s3/1.4/-1, t/2/0.45, u/2/-0.2, v/2.6/0.45, x1/0/-2.3, x2/0.8/-2.3}{
      \node[draw,circle] (\l) at (\i,\j) {} ;
    }
    \node at (1.4,-1) {$x$} ;
    \node at (0,-2.3) {\footnotesize{$x_1$}} ;
    \node at (0.8,-2.3) {\footnotesize{$x_2$}} ;
    \foreach \i/\j/\w in {{(a1) (a2)}/A/A,{(b1) (b2)}/B/B,{(A) (B)}/F/, {(ap1) (ap2)}/Ap/A',{(bp1) (bp2)}/Bp/B',{(Ap) (Bp)}/Fp/,{(s1) (s2)}/S/}{
      \node[draw,rounded corners,fit=\i,label=center:$\w$] (\j) {} ;
    }
    \node (tf) at (0.25,1.2) {$V(F)$} ;
    \node (ts) at (1.4,1.2) {$S$} ;
    \node[draw,rounded corners,fit=(tf) (ts) (S) (F) (u) (t) (s3),label=$S'$] (Sp) {} ;
    \node (tfp) at (-1.75,1.2) {$V(F')$} ;
    
    \foreach \i/\j/\k in {S/A/0,s3/S/0,s3/F/0,Sp/Ap/0,x1/Sp/0,x1/Fp/20,x2/Sp/0,x2/Fp/12}{
      \draw[thick] (\i) to [bend left=\k] (\j) ;
    }
    \foreach \i/\j/\k in {t/s1/0,t/s2/0,u/s3/0,u/t/0,x1/x2/0,v/s2/-12,v/t/0,v/x/18,v/x2/25}{
      \draw (\i) to [bend left=\k] (\j) ;
    }

    \begin{scope}[xshift=6cm, yshift=0cm]
    \foreach \l/\i/\j/\o in {s1/1.4/0.2,s2/1.4/0.7, x/1.4/-1, t/2/0.45, u/2/-0.2, v/2.6/0.45,x1/1.4/-1.8, x2/2.2/-1.8}{
      \node[draw,circle] (\l) at (\i,\j) {} ;
    }
    \node at (1.4,-1) {$x$} ;
    \node (tS) at (1.4,1.2) {$S$} ;
    \node at (1.4,-1.8) {\footnotesize{$x_1$}} ;
    \node at (2.2,-1.8) {\footnotesize{$x_2$}} ;
    \node[draw,very thick,brown,rounded corners,fit=(s1) (s2)] (S) {} ;
    \node[draw,very thick,brown,rounded corners,fit=(tS) (s1) (s2) (t) (u) (x),label=$S'$] (Sp) {} ;
    \foreach \i/\j/\k in {t/s1/0,t/s2/0,u/x/0,u/t/0,x1/x2/0,v/s2/-12,v/t/0,v/x/18,v/x2/25}{
      \draw (\i) to [bend left=\k] (\j) ;
    }
    \foreach \i/\j in {x/S,x1/Sp,x2/Sp}{
      \draw[thick] (\i) -- (\j) ;
    }
    \end{scope}
  \end{tikzpicture}
  \caption{Left: Two nested fences, where the set $X$ of the innermost fence is $\{x,x_1,x_2\} \cup V(F')$, while the set $X'$ of the outermost fence is $\{x_1,x_2\}$.
  Right: Its compact representation.}
  \label{fig:fence-nesting}
\end{figure} 

\medskip

\textbf{Constraints of the fence gadget on a 4-sequence.} 
The following lemmas are preparatory steps for the milestone that no part in a 4-sequence of $G$ can overlap $S$ (that is, intersects $S$ without containing it).

\begin{lemma}\label{lem:fg1}
  The first contraction involving two vertices of $F$ results in a vertex of red degree at least~5, except if it is a contraction of some $a_i \in A$ with some $b_j \in B$.
\end{lemma}
\begin{proof}
  We consider a pair of distinct vertices $u, v$ of $A \times A$ or $B \times B$.
  If $u$ and $v$ are non-adjacent in $F$, they both have at least three private neighbors in the total graph $\mathcal T(F)$.
  Thus their contraction results in a vertex of red degree at least~6, since by assumption these at least 6 vertices lie in distinct parts.

  If $u$ and $v$ are adjacent in $F$, they both have at least two private neighbors in $\mathcal T(F)$, and a common neighbor $w$ such that at least one of $uw$ and $vw$ is red in $F$.
  Hence their contraction results in a vertex with red degree at least~5.
\end{proof}

\begin{lemma}\label{lem:fg2}
  If the first contraction involving two vertices of $F$ is of some $a_i \in A$ with some $b_j \in B$, the red degree within $F$ of the created vertex is at least~3.
\end{lemma}
\begin{proof}
Let $u \in A$ and $v \in B$.
Then, in the total graph $\mathcal T(F)$, either $u$ and $v$ both have at least two private neighbors, or they both have at least one private neighbor and a common neighbor linked to $u$ or $v$ in red in the trigraph $F$. 
By assumption, these at least 4~neighbors or at least 3~neighbors are in distinct parts. 
Hence, in both cases, the contraction of $u$ and $v$ results in a vertex of red degree at least 3 within $F$.
\end{proof}

The last preparatory step is this easy lemma.
\begin{lemma}\label{lem:fg6}
Before a contraction involves two vertices of $V(F)$, the following holds in a partial 4-sequence:
  \begin{compactitem}
  \item no part intersects both $X$ and $S$,
  \item no part intersects both $Y$ and $S$, and
  \item no part intersects both $X$ and $Y$.
  \end{compactitem}
\end{lemma}
\begin{proof}
By~\cref{obs:fg5}, such a part would have red degree $|B|=6$, $|A|=6$, and $|A \cup B|=12$, respectively. 
\end{proof}

As a consequence we obtain the following.

\begin{lemma}\label{lem:fg3}
  In a partial 4-sequence of $G$, the first contraction involving a vertex in $V(F)$ and a vertex in $V(F) \cup S$ has to be done after $S$ is contracted into a single vertex. 
\end{lemma}
\begin{proof}
  We consider the first time a vertex $u \in V(F)$ is involved in a contraction with a~vertex of $V(F) \cup S$.
  Either (case 1) the part of $u$, $P_u$, is contracted with a part $P_v$ containing $v \in V(F)$, or (case 2) $P_u$ is contracted with a part $P$ intersecting $S$ but not $V(F)$.

  In case 1, by~\cref{lem:fg1}, $u$ and $v$ hit both $A$ and $B$.
  Thus, by~\cref{lem:fg2}, the red degree within $F$ of the resulting vertex $z$ is at least~3.
  Moreover $z$ is linked by a red edge to every part within $S$, since $S$ is fully adjacent to $A$, and fully non-adjacent to $B$.
  Thus $S$~should at this point consist of a single part.

  We now argue that case 2 is impossible in a partial 4-sequence.
  By \cref{lem:fg6}, part~$P$ cannot intersect $X \cup Y$ (nor $V(F)$, by construction).
  Thus $P \subseteq S$.
  If $u \in A$, then the contraction of $P_u$ and $P$ has incident red edges toward at least 5~vertices: three vertices of~$A$ non-adjacent to~$u$ and two private neighbors of~$u$ (in the total graph) within~$B$.
  If instead $u \in B$, the red degree of the contracted part is at least~6, as witnessed by two neighbors of $u$ in $B$, and four non-neighbors of $u$ in $A$. 
\end{proof}

We can now establish the main lemma on how a fence gadget constrains a 4-sequence.
\cref{lem:fg3,lem:fg6} have the following announced consequence: While $S$ is not contracted into a single vertex, no part within $S$ can be contracted with a part outside of $S$, and similarly vertices of $X$ cannot be contracted with vertices of $Y$.

\begin{lemma}\label{cor:fg4}
  In a partial 4-sequence, while $S$ is not contracted to a single vertex, 
  \begin{compactitem}
    \item(i) no part intersects both $S$ and $V(G) \setminus S$, nor
    \item(ii) both $X$ and $Y$.
  \end{compactitem}
\end{lemma}
\begin{proof}
  Let $H$ be a trigraph obtained by a partial 4-sequence from $G$, such that $S$ is not contained in a part of $\mathcal P(H)$.
  By~\cref{lem:fg3}, no pair of vertices in $V(F) \times (S \cup V(F))$ are in the same part of $\mathcal P(H)$ (since $S$ is not contracted to a single vertex).
  Thus we conclude by~\cref{lem:fg6}. 
\end{proof}

\medskip

\textbf{Contracting the fence gadget.}
The previous lemmas establish some constraints that the fence gadget imposes on a supposed (partial) 4-sequence.
We now see how a partial 4-sequence actually contracts a fence gadget.

Every time we are about to contract a fence gadget $F$ attached to $S$, we will ensure that the following properties hold:
\begin{compactitem}
\item no prior contraction has involved a vertex of $V(F)$,
\item no red edge has one endpoint in $V(F)$ and one endpoint outside $V(F)$, and
\item $S$ is contracted into a single vertex with red degree at most~3.
\end{compactitem}
In particular, the fence gadget $F$ still satisfies the attachment rule.

\begin{lemma}\label{lem:fg5}
  Let $H$ be a trigraph containing a fence gadget $F$ attached to a single vertex $s$ of red degree at most~3.
  We assume that $F$ respects the attachment rule in $H$.
  
  Then there is a partial 4-sequence from $H$ to $H'$, where $H'$ is the trigraph obtained from~$H$ by contracting $V(F)$ into a single vertex.
\end{lemma}

\begin{proof}
  As $F$ respects the attachment rule in $H$, every vertex of $F$ has the same (fully black) neighborhood in $V(H) \setminus (V(F) \cup \{s\})$.
  Thus, contractions within $V(F)$ will only create red edges within $V(F) \cup \{s\}$.
  We can therefore focus on the trigraph induced by $V(F) \cup \{s\}$.

  Recall the vertex labels of \cref{fig:fence-gadget}.
  We first contract $a_1$ and $b_1$.
  This creates a vertex~$c_1$ of red degree~4, and in particular it adds one to the red degree of $s$, which has now red degree at most~4 (see left-hand side of~\cref{fig:fencecontraction}).
  Now we will contract within $A$, and within $B$, not to increase more the red degree of $s$.

\begin{figure}[ht!]
  \centering
  \resizebox{400pt}{!}{
  \begin{tikzpicture}
    \foreach \i in {2,...,6}{
      \node[draw,circle,inner sep=0.03cm] (b\i) at (360 * \i/6:1) {$b_\i$} ;
      \node[draw,circle,inner sep=0.03cm] (a\i) at (360 * \i/6:2) {$a_\i$} ;
      \draw[very thick,red] (a\i) -- (b\i) ;
    }
    \node[draw,circle,inner sep=0.03cm] (a1) at (360 * 1/6:1.8) {$c_1$};
    \node[draw,circle,inner sep=0.08cm] (s) at (360 * 1/6:3.3) {$s$};
    \foreach \i [count = \ip from 3] in {2,...,5}{
      \draw (a\i) -- (a\ip) ;
      \draw (b\i) -- (b\ip) ;
      \draw[very thick,red] (a\i) -- (b\ip) ;
    }
    \draw[very thick,red] (a1) -- (a2) ;
    \draw[very thick,red] (a1) -- (b2) ;
    \draw[very thick,red] (a1) -- (b6) ;
    \draw[very thick,red] (a1) -- (s) ;
    \draw[thick] (a6) to node[above]{$e$} (a1) ;

    \foreach \i/\k in {2/0,6/0,3/-50,5/50}{
      \draw (s) to [bend left=\k] (a\i) ;
    }
    \draw (s) to [bend left=-42] (-2.1,2.1) to [bend left=-42] (a4) ;

    \begin{scope}[xshift=6cm]
      \foreach \i in {2,...,6}{
      \node[draw,circle,inner sep=0.03cm] (a\i) at (360 * \i/6:2) {$a_\i$} ;
      }
      \foreach \i in {4,...,6}{
        \node[draw,circle,inner sep=0.03cm] (b\i) at (360 * \i/6:1) {$b_\i$} ;
        \draw[very thick,red] (a\i) -- (b\i) ;
      }
    \node[draw,circle,inner sep=0.025cm] (b3) at (150:1) {\small{$b_{23}$}};
    \node[draw,circle,inner sep=0.03cm] (a1) at (360 * 1/6:1.8) {$c_1$};
    \node[draw,circle,inner sep=0.08cm] (s) at (360 * 1/6:3.3) {$s$};
    \foreach \i [count = \ip from 4] in {3,...,5}{
      \draw (a\i) -- (a\ip) ;
      \draw (b\i) -- (b\ip) ;
      \draw[very thick,red] (a\i) -- (b\ip) ;
    }
    \draw (a2) -- (a3) -- (a4) ;
    \foreach \i/\j in {a1/a2,a1/b3,a1/b6,a1/s,b3/b4,a2/b3,a3/b3}{
    \draw[very thick,red] (\i) -- (\j) ;
    }
    \draw[thick] (a6) to node[above]{$e$} (a1) ;

    \foreach \i/\k in {2/0,6/0,3/-50,5/50}{
      \draw (s) to [bend left=\k] (a\i) ;
    }
    \draw (s) to [bend left=-42] (-2.1,2.1) to [bend left=-42] (a4) ;
    \end{scope}

    \begin{scope}[xshift=12cm]
      \foreach \i in {4,...,6}{
        \node[draw,circle,inner sep=0.03cm] (a\i) at (360 * \i/6:2) {$a_\i$} ;
        \node[draw,circle,inner sep=0.03cm] (b\i) at (360 * \i/6:1) {$b_\i$} ;
        \draw[very thick,red] (a\i) -- (b\i) ;
      }
    \node[draw,circle,inner sep=0.025cm] (a3) at (150:2) {\small{$a_{23}$}};
    \node[draw,circle,inner sep=0.025cm] (b3) at (150:1) {\small{$b_{23}$}};
    \node[draw,circle,inner sep=0.03cm] (a1) at (360 * 1/6:1.8) {$c_1$};
    \node[draw,circle,inner sep=0.08cm] (s) at (360 * 1/6:3.3) {$s$};
    \foreach \i [count = \ip from 4] in {3,...,5}{
      \draw (a\i) -- (a\ip) ;
      \draw (b\i) -- (b\ip) ;
      \draw[very thick,red] (a\i) -- (b\ip) ;
    }
    \foreach \i/\j in {a1/a3,a1/b3,a1/b6,a1/s,b3/b4,a3/b3,a3/a4}{
    \draw[very thick,red] (\i) -- (\j) ;
    }
    \draw[thick] (a6) to node[above]{$e$} (a1) ;

    \foreach \i/\k in {6/0,3/-50,5/50}{
      \draw (s) to [bend left=\k] (a\i) ;
    }
    \draw (s) to [bend left=-42] (-2.1,2.1) to [bend left=-42] (a4) ;
    \end{scope}
  \end{tikzpicture}
  }
  \caption{The first three contractions of the fence gadget.}
  \label{fig:fencecontraction}
\end{figure}
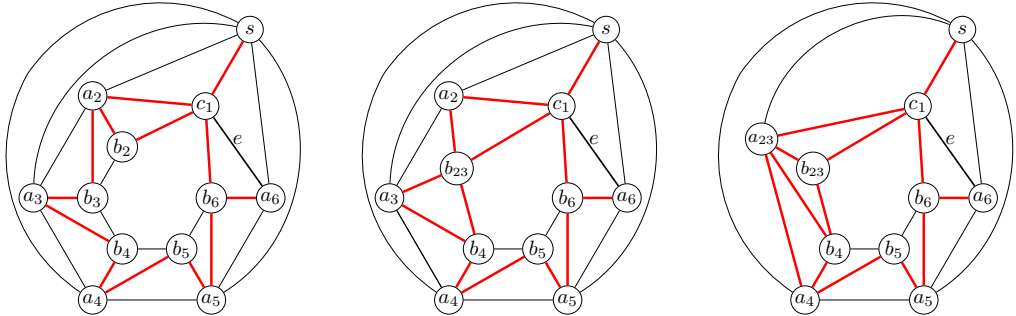

Vertices $b_2$ and $b_3$ now have only 4 neighbors in total in the fence gadget.
Thus we can contract them into $b_{23}$ and keep red degree at most~4 (see middle of~\cref{fig:fencecontraction}).
In turn, $a_2$ and $a_3$ can be contracted into $a_{23}$ for the same reason (see right of~\cref{fig:fencecontraction}).
Then we contract $b_4$ with $b_{23}$, forming $b_{234}$, and $a_4$ with $a_{23}$, forming $a_{234}$.
We contract $b_5$ and $b_{234}$ into $b_{2345}$, and $a_5$ and $a_{234}$ into $a_{2345}$.
We finally contract $b_6$ and $b_{2345}$ into $b$, and $a_6$ and $a_{2345}$ into $a$.
We contract $a$ and $c_1$, then we contract the resulting vertex and $b$.

Crucially the edge $e$ stays black until the contraction of $a_6$ and $a_{2345}$, so $c_1$ maintains a~red degree of at most~4.
Also importantly, the only contractions involving a vertex of $A$ and a vertex of $B$ are the first and last two contractions. 
Thus $s$ is incident in red to at most one vertex of the fence gadget. 
\end{proof}

\subsection{Propagation, wire, and long chain}\label{sec:wire}

A \emph{vertical set} $V$ consists of two vertices $x, y$ combined with a fence gadget $F$ attached to $\{x,y\}$.
Thus $V=\{x,y\} \cup V(F)$.
We call \emph{vertical pair} the vertices $x$ and $y$.
We will usually add a superscript to identify the different copies of vertical sets: Every vertex whose label is of the form $u^j$ belongs to the vertical set $V^j$. 

The \emph{propagation gadget} from $V^j$ to $V^{j'}$ puts all the edges between $V^j$ and $x^{j'}$ (and no other edge).
We call these edges an \emph{arc from $V^j$ to $V^{j'}$}.
We also say that the vertical set $V^{j'}$ is \emph{guarded} by $V^j$.
The pair $V^j,V^{j'}$ is said \emph{adjacent}.
Here, singleton $\{x_j'\}$ plays the role of $X$ (Definition~\ref{def:attachment-rule}) for the attachment rule of vertical set $V_j$.

The \emph{propagation digraph of $G$}, denoted by $\mathcal D(G)$, has one vertex per vertical set and an arc between two vertical sets linked by an arc (in the previous sense).
A (maximal) \emph{wire} $W$ is an induced subgraph of $G$ corresponding in $\mathcal D(G)$ to a (maximal) out-tree on at least two vertices.
See~\cref{fig:wire} for an illustration of a wire made by simply concatenating propagation gadgets.
Eventually $\mathcal D(G)$ will have out-degree at most~2, in-degree at most~2, and total degree at most~3.

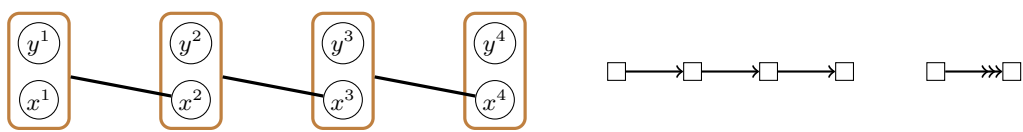
\begin{figure}[ht!]
  \centering
  \begin{tikzpicture}
    \def\k{4}
    \def\hs{2}
    \pgfmathsetmacro\xs{\k * \hs + 0.6}
    \def\vs{0.75}
    \foreach \i in {1,...,\k}{
      \node[draw,circle,inner sep=0.03cm] (a\i1) at (\i * \hs, 1 * \vs) {\small{$x^\i$}} ;
      \node[draw,circle,inner sep=0.03cm] (a\i2) at (\i * \hs, 2 * \vs) {\small{$y^\i$}} ;
      \node[draw,very thick,brown,rounded corners,fit=(a\i1) (a\i2)] (b\i2) {} ;
    }
    \foreach \i [count = \im from 1] in {2,...,\k}{
      \draw[very thick] (a\i1) -- (b\im2) ;
    }

    \begin{scope}[xshift= \xs cm]
    \foreach \i in {1,...,\k}{
      \node[draw] (v\i) at (0.5 * \i * \hs, 1.5 * \vs) {} ;
    }
    \foreach \i [count = \im from 1] in {2,...,\k}{
      \draw[thick, ->] (v\im) -- (v\i) ;
    }

    \node[draw] (y) at (0.5 * 5.2 * \hs, 1.5 * \vs) {} ;
    \node[draw] (z) at (0.5 * 6.2 * \hs, 1.5 * \vs) {} ;
    \draw[thick, ->>>] (y) -- (z) ;
    \end{scope}
    
  \end{tikzpicture}
  \caption{Left: A non-branching wire made of 4 vertical sets and 3 propagation gadgets. Every vertical set is guarded by the vertical set just to its left. Center: A more compact representation, which corresponds to the propagation digraph. Right: Symbolic representation of the long chain, that is, of the represented wire if $L=4$.}
  \label{fig:wire}
\end{figure}

The \emph{children} of a vertical set $V$ are the vertical sets that $V$ guards.
The \emph{root of wire $W$} is the unique vertical set of in-degree 0 in $\mathcal D(W)$.
The \emph{leaves of wire $W$} are the vertical sets of out-degree 0 in $\mathcal D(W)$. 
A wire is said \emph{primed} when the vertical pair of its root has been contracted.

A wire $W$ is \emph{non-branching} if every vertex of $\mathcal D(W)$ has out-degree at most 1; hence, $\mathcal D(W)$ is a directed path.
A \emph{long chain} is a wire $W$ such that $\mathcal D(W)$ is a directed path on $L$ vertices, where integer $L$ will be specified later (and can be thought as logarithmic in the total number of fences which do not belong to vertical sets). 
Otherwise, if $\mathcal D(W)$ has at least one vertex with out-degree at least 2, wire $W$ is said~\emph{branching}.
A~vertical set with two children is also said~\emph{branching}.
See~\cref{fig:branching-wire} for an example of a branching wire with exactly one branching vertical set.

\medskip

\textbf{Constraints of the propagation gadget on a 4-sequence.}
We provide the proof that a~contraction in a vertical set $V$ is only possible when the vertical pair of all the vertical sets $V'$ with a directed path to $V$ in $\mathcal D(G)$ has been contracted. 

\begin{lemma}\label{lem:propagation}
  Let $V^j$ and $V^{j'}$ be two vertical sets with an arc from $V^j$ to $V^{j'}$.
  In a partial 4-sequence from $G$, any contraction involving two vertices of $V^{j'}$ has to be preceded by the contraction of $x^j$ and $y^j$. 
\end{lemma}
\begin{proof}
  We recall the notations of~\cref{subsection:fg}. 
  Let $F$ be the fence gadget attached to $S=\{x^j,y^j\}$, $X$ the neighborhood of $V(F)$ outside of $S \cup V(F)$, and $Y$ the vertices that are not in $S \cup V(F) \cup X$.
  (We always assume that the attachment rule is satisfied.)
  We have $x^{j'} \in X$ and $y^{j'} \in Y$, therefore by the second item of~\cref{cor:fg4}, their contraction has to be preceded by the contraction of $x^j$ and $y^j$.
  Now applying the first item of~\cref{cor:fg4} to the fence gadget $F'$ attached to $S'=\{x^{j'},y^{j'}\}$, and~\cref{lem:fg3}, any contraction involving a pair of $V^{j'}$ distinct from $S'$ has to be preceded by the contraction of $x^{j'}$ and $y^{j'}$.
\end{proof}

\begin{figure}[ht!]
  \centering
  \begin{tikzpicture}
    \def\k{4}
    \def\hs{2}
    \pgfmathsetmacro\xs{\k * \hs + 1}
    \def\vs{0.75}
    \foreach \i in {1,...,\k}{
      \node[draw,circle,inner sep=0.03cm] (a\i1) at (\i * \hs, 1 * \vs) {\small{$x^\i$}} ;
      \node[draw,circle,inner sep=0.03cm] (a\i2) at (\i * \hs, 2 * \vs) {\small{$y^\i$}} ;
      \node[draw,very thick,brown,rounded corners,fit=(a\i1) (a\i2)] (b\i2) {} ;
    }
    \foreach \i [count = \im from 1] in {2,...,\k}{
      \draw[very thick] (a\i1) -- (b\im2) ;
    }
    \foreach \i/\l in {3/5,4/6}{
      \node[draw,circle,inner sep=0.03cm] (a\l1) at (\i * \hs, 3.4 * \vs) {\small{$x^\l$}} ;
      \node[draw,circle,inner sep=0.03cm] (a\l2) at (\i * \hs, 4.4 * \vs) {\small{$y^\l$}} ;
      \node[draw,very thick,brown,rounded corners,fit=(a\l1) (a\l2)] (b\l2) {} ;
    }
    \draw[very thick] (a61) -- (b52) ;
    \draw[very thick] (a51) -- (b22) ;

    \begin{scope}[xshift= \xs cm]
    \foreach \i in {1,...,\k}{
      \node[draw] (v\i) at (0.5 * \i * \hs, 2.2 * \vs) {} ;
    }
    \node[draw] (v5) at (0.5 * 3 * \hs, 3.2 * \vs) {} ;
    \node[draw] (v6) at (0.5 * 4 * \hs, 3.2 * \vs) {} ;
    \foreach \i [count = \im from 1] in {2,...,\k}{
      \draw[thick, ->] (v\im) -- (v\i) ;
    }
    \draw[thick, ->] (v2) -- (v5) ;
    \draw[thick, ->] (v5) -- (v6) ;
    \end{scope} 
  \end{tikzpicture}
  \caption{An example of a branching wire with its corresponding propagation digraph.}
  \label{fig:branching-wire}
\end{figure}
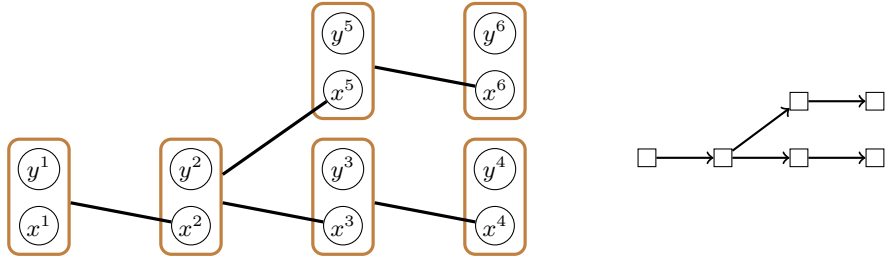 

Henceforth, when we say that a contraction is \emph{preceded} by another contraction, it includes the case that the two contractions are in fact the same.
By a straightforward induction, we obtain the following from~\cref{lem:propagation} (and~\cref{lem:fg3}). 

\begin{lemma}\label{cor:wire1}
  In a partial 4-sequence from $G$, any contraction involving a pair of vertices in a~vertical set $V$ has to be preceded by the contraction of the vertical pair of every vertical set $V'$ such that there is a directed path from $V'$ to $V$ in $\mathcal D(G)$.
\end{lemma}

\textbf{Contracting wires.}
As the roots and leaves of wires will be connected to other gadgets, we postpone the description of how to contract wires until after building the overall construction~$G$.
Intuitively though, contracting a wire (in the vacuum) consists of contracting the vertical pair of its root, then its fence gadget by applying~\cref{lem:fg5}, and finally recursively contracting the subtrees rooted at its children.
Since~$\mathcal D(G)$ has total degree at most~3, every vertex has red degree at most~4 (for the at most~3 adjacent vertical sets, plus the pendant vertex of the fence gadget).

\subsection{Binary AND gate}

The \emph{binary} AND \emph{gate} (AND gadget, for short) simply consists of three vertical sets $V^1, V^2, V^3$ with an arc from $V^1$ to $V^3$, and an arc from $V^2$ to $V^3$.
As usual, the vertical pairs of $V^1$, $V^2$, and $V^3$, are $\{x^1,y^1\}$, $\{x^2,y^2\}$, and $\{x^3,y^3\}$, respectively.
We call the vertical sets $V^1, V^2$ the \emph{inputs} of the AND gadget, and the vertical set $V^3$ the \emph{output} of the AND gadget.
See~\cref{fig:and-gadget} for an illustration.

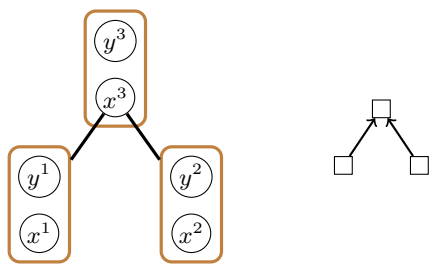
\begin{figure}[ht!]
  \centering
  \begin{tikzpicture}
    \def\hs{2}
    \pgfmathsetmacro\xs{2 * \hs + 1}
    \def\vs{0.75}
    \foreach \i in {1,2}{
      \node[draw,circle,inner sep=0.03cm] (a\i1) at (\i * \hs, 1 * \vs) {\small{$x^\i$}} ;
      \node[draw,circle,inner sep=0.03cm] (a\i2) at (\i * \hs, 2 * \vs) {\small{$y^\i$}} ;
      \node[draw,very thick,brown,rounded corners,fit=(a\i1) (a\i2)] (b\i) {} ;
    }

    \node[draw,circle,inner sep=0.03cm] (x) at (1.5 * \hs, 3.4 * \vs) {\small{$x^3$}} ;
    \node[draw,circle,inner sep=0.03cm] (y) at (1.5 * \hs, 4.4 * \vs) {\small{$y^3$}} ;
    \node[draw,very thick,brown,rounded corners,fit=(x) (y)] (b) {} ;
    
    \draw[very thick] (b1) -- (x) ;
    \draw[very thick] (b2) -- (x) ;

    \begin{scope}[xshift= \xs cm]
    \foreach \i in {1,2}{
      \node[draw] (v\i) at (0.5 * \i * \hs, 2.2 * \vs) {} ;
    }
    \node[draw] (v) at (0.5 * 1.5 * \hs, 3.2 * \vs) {} ;
    
    \draw[thick, ->] (v1) -- (v) ;
    \draw[thick, ->] (v2) -- (v) ;
    \end{scope}
  \end{tikzpicture}
  \caption{An AND gadget, and the corresponding propagation digraph.}
  \label{fig:and-gadget}
\end{figure} 

\medskip

\textbf{Constraint of the AND gadget on a 4-sequence.}
By~\cref{lem:propagation}, we readily derive:

\begin{lemma}\label{cor:and1}
  Assume $G$ contains an AND gadget with inputs $V^1, V^2$, and output $V^3$. 
  In a~partial 4-sequence from $G$, any contraction involving two vertices of $V^3$ has to be preceded by the contraction of $x^1$ and $y^1$, and the contraction of $x^2$ and $y^2$. 
\end{lemma}

\textbf{Contraction of an AND gadget.}
Once $V^1$ and $V^2$ are contracted into single vertices, one can contract the vertical pair $x^3,y^3$.
This results in a vertex of red degree~2.
Thus one can contract the fence gadget of $V^3$ by applying~\cref{lem:fg5}.

\subsection{Binary OR gate}\label{sec:or-gate}

The \emph{binary} OR \emph{gate} (OR \emph{gadget}) is connected to three vertical sets: two \emph{inputs} $V^1, V^2$, and one \emph{output} $V^3$.
We start by building two vertical sets $V, V'$ whose vertical pairs are $\{a,b\}$ and $\{c,d\}$, respectively.
The edges $ac$ and $bd$ are added, as well as a vertex~$e$ adjacent to $a$ and to $c$.
Finally a fence is attached to $\{e\} \cup V \cup V'$.

The OR gadget is connected to its inputs and output, in the following way.
Vertex $a$ is made adjacent to $x^1$ and to $y^1$ (but not to their fence gadget).
Similarly vertex $b$ is linked to $x^2$ and $y^2$.
Finally $x^3$ is adjacent to all the vertices of the OR gadget, that is, $\{e\} \cup V \cup V'$ plus the vertices of the outermost fence.
See~\cref{fig:OR} for a representation of the OR gadget. 

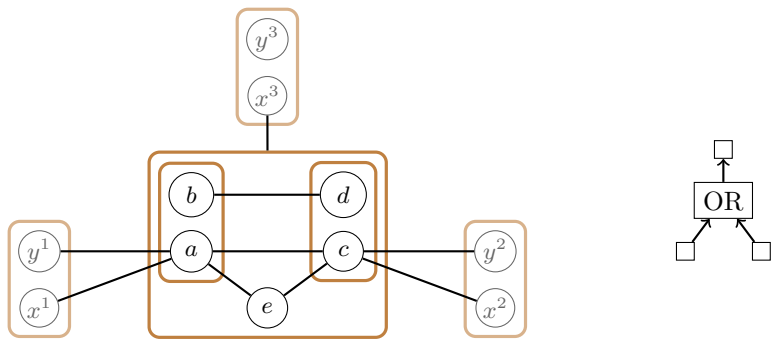
\begin{figure}[ht!]
  \centering
  \begin{tikzpicture}
    \def\hs{2}
    \def\vs{0.75}
    \def\op{0.6}
  
  \node[draw,circle,inner sep=0.03cm,opacity=\op] (v11) at (0, -\vs) {\small{$x^1$}};
  \node[draw,circle,inner sep=0.03cm,opacity=\op] (v21) at (0,0) {\small{$y^1$}};
  \node[draw,circle] (u11) at (\hs , 0) {\small{$a$}};
  \node[draw,circle] (u21) at (\hs ,\vs) {\small{$b$}};
  \node[draw,circle,inner sep=0.03cm,opacity=\op] (v12) at (3 * \hs , -\vs) {\small{$x^2$}};
  \node[draw,circle,inner sep=0.03cm,opacity=\op] (v22) at (3 * \hs ,0) {\small{$y^2$}};
  \node[draw,circle] (u12) at (2 * \hs , 0) {\small{$c$}};
  \node[draw,circle] (u22) at (2 * \hs ,\vs) {\small{$d$}};
  \node[draw,circle] (t) at (1.5 * \hs ,-\vs) {\small{$e$}};

  \node[draw,circle,inner sep=0.03cm,opacity=\op] (v13) at (1.5 * \hs, 2.75 * \vs) {\small{$x^3$}};
  \node[draw,circle,inner sep=0.03cm,opacity=\op] (v23) at (1.5 * \hs, 3.75 * \vs) {\small{$y^3$}};
  
  \node[draw,very thick,brown,rounded corners,opacity=\op, fit= (v11) (v21) ] (V1) {} ;
  \node[draw,very thick,brown,rounded corners,fit= (u11) (u21) ] (b) {} ;
  \node[draw,very thick,brown,rounded corners,opacity=\op, fit= (v12) (v22) ] (V2) {} ;
  \node[draw,very thick,brown,rounded corners,fit= (u12) (u22) ] (c) {} ;
  \node[draw,very thick,brown,rounded corners,fit= (u21) (u22) (t) (b) (c) ] (o) {} ;
  \node[draw,very thick,brown,rounded corners,opacity=\op, fit= (v13) (v23) ] (V3) {} ;
  
  \foreach \j in {1,...,2}{
    \draw[thick] (v\j1) -- (u11) ;
    \draw[thick] (v\j2) -- (u12) ;
    \draw[thick] (u\j2) -- (u\j1) ;
  }
  \draw[thick] (u12) -- (t) ;
  \draw[thick] (t) -- (u11) ;
  \draw[thick] (o) -- (v13) ;

  \begin{scope}[xshift=8.5cm]
    \foreach \l/\i/\j in {a/0/0,b/2/0,c/1/1.8}{
      \node[draw] (\l) at (0.25 * \i * \hs, \j * \vs) {} ;
    }
    \node[draw] (og) at (0.25 * \hs, 0.9 * \vs) {OR} ;
    \foreach \i/\j in {a/og,b/og,og/c}{
      \draw[thick,->] (\i) -- (\j) ;
    }
  \end{scope}
  
  \end{tikzpicture}
  \caption{An OR gadget attached to inputs $V^1, V^2$ and output $V^3$, and its symbolic representation.
    These vertical sets are technically not part of the OR gate, so we represent them slightly dimmer.   
    The outermost fence of the OR gadget can only be contracted when at least one the pairs $x^1, y^1$ and $x^2, y^2$ have been contracted.
    Only after that, can $x^3$ and $y^3$ be contracted together.}
    \label{fig:OR}
\end{figure}  

\textbf{Constraint of the OR gadget on a 4-sequence.}
By design, one can only start contracting the OR gadget after the vertical pair of at least one of its two inputs has been contracted.
This implies that no contraction can involve $V^3$ before at least one the vertical pairs $\{x^1,y^1\}$ and $\{x^2,y^2\}$ is contracted.

\begin{lemma}\label{lem:or1prep}
  Assume $G$ contains an OR gadget attached to inputs $V^1$ and $V^2$.
  In a partial 4-sequence from $G$, the contractions of $a, b$ and of $c, d$ have to be preceded by the contraction of $x^1, y^1$ or the contraction of $x^2, y^2$.
\end{lemma}

\begin{proof}
  Assume none of the pairs $\{a,b\}, \{c,d\}, \{x^1,y^1\}, \{x^2,y^2\}$ have been contracted.
  Because of the fences, by~\cref{cor:fg4}, all the vertices $x^1, y^1, a, b, c, d, x^2, y^2$ and $e$ are in distinct parts.
  Therefore contracting $a$ and $b$ would create a vertex of red degree at least $5$, considering the (singleton) parts of $x^1, y^1, c, d, e$.
  Symmetrically contracting $c$ and $d$ yields at least five red neighbors, considering the (singleton) parts of $x^2, y^2, a, b, e$. 
\end{proof}

From \cref{lem:or1prep} we get the following.
\begin{lemma}\label{lem:or1}
  Assume $G$ contains an OR gadget attached to inputs $V^1$ and $V^2$.
  In a partial 4-sequence from $G$, no contraction involving a vertex of $V^3$ can happen before either the pair $x^1, y^1$ or the pair $x^2, y^2$ is contracted.
\end{lemma}
\begin{proof}
  Suppose neither $x^1, y^1$ nor $x^2, y^2$ is contracted.
  By~\cref{lem:or1prep}, the pairs $\{a,b\}$ and $\{c,d\}$ cannot be contracted.
  By the first item of~\cref{cor:fg4}, no contraction can involve a~vertex of $\{a,b,c,d,e\}$.
  As $x^3$ is adjacent to all these vertices but not $y^3$, one cannot contract the vertical pair $\{x^3,y^3\}$.
  Hence by~\cref{lem:fg3} no contraction can involve a vertex of $V^3$.
\end{proof}

\textbf{Contraction of the OR gadget.}
We now show how to contract the OR gate when the vertical pair of one of its inputs has been contracted.

\begin{lemma} \label{lem:or2}
  Assume that $x^1$ and $y^1$ have been contracted into $z^1$, and that $z^1$, $x^2$, and $y^2$ all have red degree at most 3.
  Then there is a partial 4-sequence that contracts the whole OR gadget to a single vertex with only three red neighbors: $z^1$, $x^2$, and $y^2$.
  (The same holds symmetrically if $x^2$ and $y^2$ have been contracted into a single vertex.)
\end{lemma}

\begin{proof}
  First contract $a$ and $b$ into vertex $\alpha$ of red degree 4.
  At this point the fence of $\{a,b\}$ cannot be contracted yet, as this would make the red degree of $\alpha$ go above 4.
  Hence we next contract $c$ and $d$ into $\gamma$, decreasing the red degree of $\alpha$ to 3.
  Note that $\gamma$ has only 4 red neighbors: $\alpha, e, x^2, y^2$.
  
  By~\cref{lem:fg5}, we can now contract the fence gadget of $\{a,b\}$ to a single vertex.
  Next we contract this latter vertex with $\alpha$, and the resulting vertex with $t$; we call $\alpha'$ the obtained vertex.
  Now $\gamma$ has only red degree~3, so we can contract the fence gadget of $\{c,d\}$ to a single vertex that we further contract with $\gamma$; we call $\gamma'$ the obtained vertex.
  We contract $\alpha'$ and $\gamma'$ in a vertex $\varepsilon$ of red degree 3; its three red neighbors are $z^1$, $x^2$, and $y^2$.
  Again by~\cref{lem:fg5}, the outermost fence of the OR gadget can be contracted into a single vertex, that we finally contract with $\varepsilon$.
  This results in a vertex with three red neighbors: $z^1$, $x^2$, and $y^2$.
 
  Throughout this process the red degree of $z^1$, $x^2$, and $y^2$ never goes above 4.
  Indeed the red degree of these vertices is initially at most 3, while they have exactly one black neighbor in the entire OR gadget (so at most one part to be in conflict with).
\end{proof} 

\subsection{Variable gadget}

We first describe the variable gadget for, say, a variable $x$.
We start by attaching a fence on a set formed by three vertices: $x, \top, \bot$.
We add two disjoint copies of an OR gadget, with vertices $a^\top,b^\top,c^\top,d^\top,e^\top$ (resp.~$\{a^\bot,b^\bot,c^\bot,d^\bot,e^\bot\}$) plus the vertices in the fence gadgets; see~\cref{sec:or-gate,fig:OR}.
We call $T'$, respectively~$U'$, the vertex sets of the two OR gadgets.
We link $a^\top$ to~$x$ and to~$\top$, and we link $a^\bot$ to~$x$ and to~$\bot$.
We then add a vertex $f^\top$ (resp.~$f^\bot$), make it adjacent to~$c^\top$ (resp.~$c^\bot$), and add a fence $F'^\top$ (resp.~$F'^\bot$) attached to $T' \cup \{f^\top\}$ (resp.~$U' \cup \{f^\bot\}$).
We add another vertex $g^\top$ (resp.~$g^\bot$), make it adjacent to~$c^\top$ (resp.~$c^\bot$), and add a fence $F^\top$ (resp.~$F^\bot$) attached to $T = T' \cup V(F'^\top) \cup \{f^\top, g^\top\}$ (resp.~$U = U' \cup V(F'^\bot) \cup \{f^\bot, g^\bot\}$).

The variable gadget is \emph{half-guarded} by two vertical sets $V^1, V^2$ (with vertical pairs $x^1, y^1$ and $x^2, y^2$):
We make $\top$ adjacent to $x^1$ and $y^1$, and we make $\bot$ adjacent to $x^2$ and $y^2$.
Finally $T$ guards a vertical set $V^3$ (with vertical pair $x^3, y^3$), and $U$ guards a vertical set $V^4$ (with vertical pair $x^4, y^4$), that is, $x^3$ is fully adjacent to $T$, and $x^4$ is fully adjacent to $U$.
Vertical set $V^3$ is the root of what we call the \emph{wire of literal $+x$}, which is the maximal wire containing $V^3$.
We add the $+$ to differentiate the \emph{literal} from the \emph{variable} $x$.
Similarly, $V^4$ is the root of the \emph{wire of literal $\neg x$}, that is, the maximal wire containing $V^4$.

This finishes the construction of the variable gadget; see~\cref{fig:VG} for an illustration.

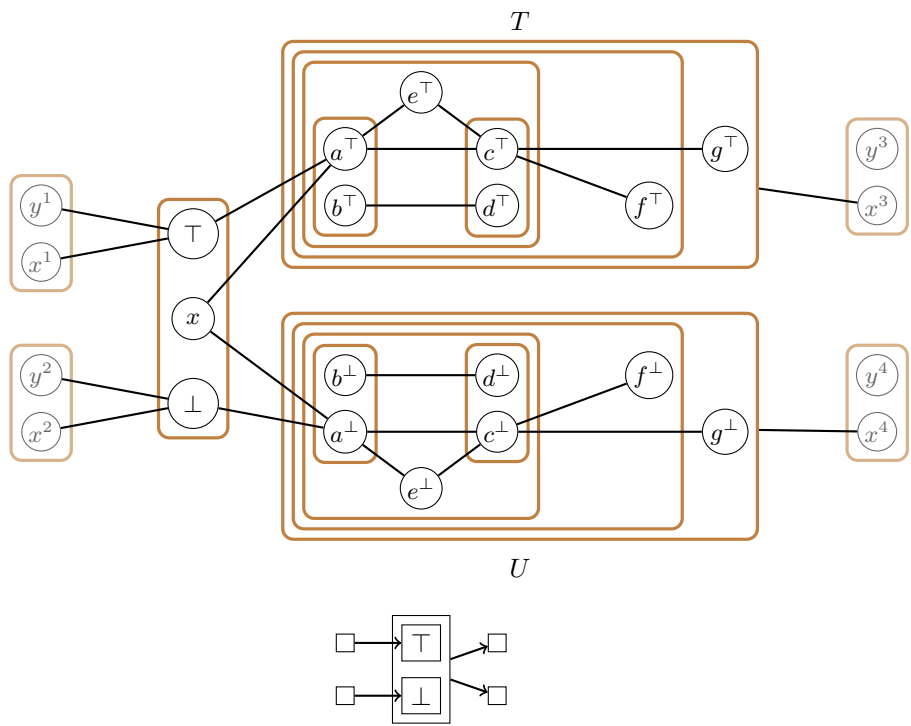
\begin{figure}[ht!]
  \centering
  \begin{tikzpicture}
    \def\hs{2}
    \def\vs{0.75}
    \def\op{0.6}
  
  \node[draw,circle,opacity=\op,inner sep=0.03cm] (v11) at (0, \vs) {\small{$x^1$}};
  \node[draw,circle,opacity=\op,inner sep=0.03cm] (v21) at (0, 2* \vs) {\small{$y^1$}};
  \node[draw,circle,opacity=\op,inner sep=0.03cm] (v12) at (0 , -2 *\vs) {\small{$x^2$}};
  \node[draw,circle,opacity=\op,inner sep=0.03cm] (v22) at (0 ,-\vs) {\small{$y^2$}};
  \node[draw,circle] (false) at (1 * \hs , - 1.5 * \vs) {\small{$\bot$}};
  \node[draw,circle] (x) at (1 * \hs ,0) {\small{$x$}};
  \node[draw,circle] (true) at (1 * \hs ,1.5 *\vs) {\small{$\top$}};
  \node[draw,circle,inner sep=0.03cm] (u11) at (2 * \hs , -2 * \vs) {\small{$a^\bot$}};
  \node[draw,circle,inner sep=0.03cm] (u21) at (2 * \hs ,-\vs) {\small{$b^\bot$}};
  \node[draw,circle,inner sep=0.03cm] (u5) at (4.5 * \hs , -2 *\vs) {\small{$g^\bot$}};
  \node[draw,circle,inner sep=0.03cm] (u4) at (4 * \hs ,-\vs) {\small{$f^\bot$}};
  \node[draw,circle,inner sep=0.03cm] (u12) at (3 * \hs , -2 *\vs) {\small{$c^\bot$}};
  \node[draw,circle,inner sep=0.03cm] (u22) at (3 * \hs ,-\vs) {\small{$d^\bot$}};
  \node[draw,circle,inner sep=0.03cm] (u3) at (2.5 * \hs ,-3 *\vs) {\small{$e^\bot$}};
  \node[draw,circle,inner sep=0.03cm] (t11) at (2 * \hs , 3 * \vs) {\small{$a^\top$}};
  \node[draw,circle,inner sep=0.03cm] (t21) at (2 * \hs , 2 *\vs) {\small{$b^\top$}};
  \node[draw,circle,inner sep=0.03cm] (t5) at (4.5 * \hs , 3 *\vs) {\small{$g^\top$}};
  \node[draw,circle,inner sep=0.03cm] (t4) at (4 * \hs ,2*\vs) {\small{$f^\top$}};
  \node[draw,circle,inner sep=0.03cm] (t12) at (3 * \hs , 3 *\vs) {\small{$c^\top$}};
  \node[draw,circle,inner sep=0.03cm] (t22) at (3 * \hs ,2 *\vs) {\small{$d^\top$}};
  \node[draw,circle,inner sep=0.03cm] (t3) at (2.5 * \hs ,4 *\vs) {\small{$e^\top$}};

  \node[draw,circle,opacity=\op,inner sep=0.03cm] (v13) at (5.5 * \hs, 2 * \vs) {\small{$x^3$}};
  \node[draw,circle,opacity=\op,inner sep=0.03cm] (v23) at (5.5 * \hs, 3 * \vs) {\small{$y^3$}};
  \node[draw,circle,opacity=\op,inner sep=0.03cm] (v14) at (5.5 * \hs, -2 * \vs) {\small{$x^4$}};
  \node[draw,circle,opacity=\op,inner sep=0.03cm] (v24) at (5.5 * \hs, -1 * \vs) {\small{$y^4$}};
  \node[draw,very thick,brown,rounded corners,opacity=\op,fit= (v13) (v23) ] (V3) {} ;
  \node[draw,very thick,brown,rounded corners,opacity=\op,fit= (v14) (v24) ] (V4) {} ;
  
  \node[draw,very thick,brown,rounded corners,fit= (u11) (u21) ] (b) {} ;
  \node[draw,very thick,brown,rounded corners,fit= (u12) (u22) ] (c) {} ;
  \node[draw,very thick,brown,rounded corners,fit= (u21) (u22) (u3) (b) (c) ] (g) {} ;
  \node[draw,very thick,brown,rounded corners,fit= (u21) (u22) (g) (u4) ] (h) {} ;
  \node[draw,very thick,brown,rounded corners,fit= (u21) (u22) (h) (u5) ] (U) {} ;
  \node[yshift=-2.5* \vs cm] at (U) {$U$};
  \node[draw,very thick,brown,rounded corners,fit= (t11) (t21) ] (bb) {} ;
  \node[draw,very thick,brown,rounded corners,fit= (t12) (t22) ] (cc) {} ;
  \node[draw,very thick,brown,rounded corners,fit= (t21) (t22) (t3) (bb) (cc) ] (gg) {} ;
  \node[draw,very thick,brown,rounded corners,fit= (t21) (t22) (gg) (t4) ] (hh) {} ;
  \node[draw,very thick,brown,rounded corners,fit= (t21) (t22) (hh) (t5),label=$T$ ] (T) {} ; 
  \node[draw,very thick,brown,rounded corners,opacity=\op,fit= (v11) (v21) ] {} ;
  \node[draw,very thick,brown,rounded corners,opacity=\op,fit= (v12) (v22) ]  {} ;
  \node[draw,very thick,brown,rounded corners,fit= (true) (false) ]  {} ;

    \foreach \j in {1,...,2}{
      \draw[thick] (u\j2) -- (u\j1) ;
    }
    \draw[thick] (u12) -- (u3) ;
    \draw[thick] (u3) -- (u11) ;
    \draw[thick] (u4) -- (u12) ;
    \draw[thick] (u5) -- (u12) ;

    \draw[thick] (T) -- (v13) ;
    \draw[thick] (U) -- (v14) ;
    
    \foreach \j in {1,...,2}{
    \draw[thick] (t\j2) -- (t\j1) ;
    }
    \draw[thick] (t12) -- (t3) ;
    \draw[thick] (t3) -- (t11) ;
    \draw[thick] (t4) -- (t12) ;
    \draw[thick] (t5) -- (t12) ;
  
    \draw[thick] (v11) -- (true) ;
    \draw[thick] (v21) -- (true) ;
    \draw[thick] (v22) -- (false) ;
    \draw[thick] (v12) -- (false) ;
    \draw[thick] (t11) -- (true) ;
    \draw[thick] (u11) -- (false) ;
    \draw[thick] (t11) -- (x) ;
    \draw[thick] (u11) -- (x) ;

    \begin{scope}[yshift=-5cm, xshift=4cm]
      \foreach \l/\i/\j/\w in {a/0/0.7/, b/0/0/, c/2/0.7/, d/2/0/, t/1/0.7/{$\top$}, f/1/0/{$\bot$}}{
        \node[draw] (\l) at (\i,\j) {\w} ;
      }
      \node[draw,fit=(t) (f)] (vg) {} ;
      \foreach \i/\j in {a/t, b/f, vg/c, vg/d}{
        \draw[thick,->] (\i) -- (\j) ;
      }
    \end{scope}
    
  \end{tikzpicture}
  \caption{Top: A variable gadget half-guarded by $V^1$, $V^2$, and with outputs $V^3, V^4$.
  Bottom: Its compact representation.} 
    \label{fig:VG}
\end{figure}  

\medskip

\textbf{Constraints of the variable gadget on a 4-sequence.}
Because of the fence gadget attached to $\{\top, x, \bot\}$, one has at some point to contract $\top$ and $\bot$ (be it with or without $x$).
A first observation is that this has to wait that the vertical pair of $V^1$ or $V^2$ is contracted.

\begin{lemma}\label{lem:VG1}
  Assume $G$ has a variable gadget half-guarded by vertical sets $V^1$ and $V^2$.
  In a~partial 4-sequence from $G$, the contraction of $\top$ and $\bot$ has to be preceded by the contraction of the pair $x^1, y^1$ or of the pair $x^2, y^2$.
\end{lemma}

\begin{proof}
  The pairs $\{x^1, y^1\}$, $\{x^2, y^2\}$, and $\{\top,\bot\}$ are contained in three disjoint sets $S^1, S^2, S$, respectively, to which a fence is attached.
  Thus before any of these pairs are contracted, by~\cref{cor:fg4}, a vertex outside $S^1 \cup S^2 \cup S$, like $a^\bot$, is in a different part than the six vertices $x^1, y^1, x^2, y^2, \top, \bot$.
  Therefore contracting $\top$ and $\bot$ would create a vertex with five red neighbors, considering the parts of $x^1, y^1, x^2, y^2, a^\bot$.
\end{proof}

We next show that no contraction is possible in $U$ (resp.~in $T$), while $x$ and $\bot$ (resp.~$x$ and $\top$) are not contracted.

\begin{lemma}\label{lem:VG2}
  In a partial 4-sequence, the contractions of $a^\bot$ and $b^\bot$ (resp.~$a^\top$ and $b^\top$) and of $c^\bot$ and $d^\bot$ (resp.~$c^\top$ and $d^\top$) have to be preceded by the contraction of $x$ and $\bot$ (resp.~$x$ and $\top$).
  Therefore no contraction is possible in $U$ (resp.~$T$) before $x$ and $\bot$ (resp.~$x$ and $\top$) are contracted.
\end{lemma}

\begin{proof}
  Since the two statements are symmetric, we only show it with $\bot$.
  Assume none of the pairs $\{x,\bot\}$, $\{a^\bot, b^\bot\}$, $\{c^\bot, d^\bot\}$ are contracted.
  Because of the fence gadgets, by the first item of~\cref{cor:fg4}, the vertices $x, \bot, a^\bot, b^\bot, c^\bot, d^\bot, e^\bot, f^\bot$ are pairwise in distinct parts.
  Therefore contracting $a^\bot$ and $b^\bot$ or $c^\bot$ and $d^\bot$ would create a vertex of red degree at least~5.
  The structure of the fence gadgets in $U$ thus prevents any contraction.
\end{proof}

We deduce that priming the wire of $\neg x$ (resp. $+x$) can only be done after $x$ and $\bot$ (resp.~$x$ and $\top$) are contracted.

\begin{lemma}\label{lem:VG3}
  Assume that $G$ has a variable gadget with outputs the vertical sets $V^3$ (root of the wire of $+x$) and $V_4$ (root of the wire of $\neg x$).
  In a partial 4-sequence from $G$, the contraction of $x^4$ and $y^4$ (resp.~$x^3$ and $y^3$) has to be preceded by the contraction of $x$ and $\bot$ (resp.~$x$ and $\top$).
\end{lemma}
\begin{proof}
  By the second item of~\cref{cor:fg4} applied to the fence attached to $U$, the pair $x^4, y^4$ cannot be contracted until $U$ is not contracted to a single vertex.
  Thus by~\cref{lem:VG2}, the pair $x^4, y^4$ can only be contracted after the pair $x, \bot$ is contracted. 
  The other statement is obtained symmetrically.
\end{proof}

\textbf{Contraction of the variable gadget.}
We show two options to contract a ``half'' of the variable gadget, either $T$ and its fence, or $U$ and its fence, into a~single vertex. 

\begin{lemma}\label{lem:VG3}
  There is a partial 4-sequence that contracts $x$ and $\top$ together, and $T \cup F^\top$ into a~single vertex.
  Symmetrically there is a partial 4-sequence that contracts $x$ and $\bot$ together, and $U \cup F^\bot$ into a~single vertex.
\end{lemma}

\begin{proof}
  We first contract $x$ and $\top$ into a vertex $+x$.
  Observe that $+x$ has exactly three red neighbors: $x^1$, $y^1$, and $a^\bot$.
  Thus $\{a^\top, b^\top, c^\top, d^\top, e^\top\}$ and their three fences can be contracted exactly like an OR gadget.
  So by~\cref{lem:or2}, there is a partial 4-sequence that contracts all these vertices to a single vertex $u$, with three red neighbors ($+x$, $f^\top$, and $g^\top$).
  We can now contract $u$ and $f^\top$ into $u'$, followed by contracting their fence gadget $F'^\top$ into a~single vertex, by~\cref{lem:fg5}.
  That pendant vertex can be contracted to $u'$, and the result to $g^\top$, forming vertex $v$.
  Finally, again by~\cref{lem:fg5}, the fence $F^\top$ attached to $T$ can be contracted into a~single vertex, which can be contracted with $v$.
 
  The other sequence is the symmetric.
\end{proof}

We now see how the second ``half'' of the variable gadget can be contracted once the vertical pairs of the \emph{half-guards} $V^1, V^2$ have been contracted.  

\begin{lemma}\label{lem:VG4}
  Assume $T \cup F^\top$ (resp.~$U \cup F^\bot$) has been contracted into a single vertex $u$, and that the pairs $\{\top,x\}$ (resp.~$\{\bot,x\}$), $\{x^1,y^1\}$, and $\{x^2,y^2\}$ have been contracted into $+x$ (resp.~$\neg x$), $z^1$, and $z^2$, respectively.
  We further assume that the red degree of $z^2$ (resp.~$z^1$) is at most~3.
  Then there is partial 4-sequence that contracts $\top, x, \bot$ and their fence into a~single vertex, and $U \cup F^\bot$ (resp.~$T \cup F^\top$) into a~single vertex.
\end{lemma}

\begin{proof}
  We contract $\bot$ with $+x$ into $v$ of red degree~4.
  This increases the red degree of $z^2$ by one, which remains at most~4.
  We then contract $U \cup F^\bot$ into a single vertex $w$, like in~\cref{lem:VG3}.
  We contract $u$ and $w$ into $y$, a vertex of red degree at most~3.
  Now $v$ has degree~3.
  So we can contract its fence gadget by~\cref{lem:fg5}.
  We further contract $v$ with its pendant neighbor, and finally with $y$.
  What results is a unique vertex with four red neighbors.

  The other partial sequence is symmetric.
\end{proof}

\subsection{Clause gadget}

A \emph{3-clause gadget} (or simply \emph{clause gadget}) has for \emph{inputs} three vertical sets $V^1, V^2, V^3$, and for \emph{output} one vertical set $V^4$.
It consists of combining two OR gadgets, and using long chains to make the OR gadgets \emph{distant enough}.
We add an OR gadget with input $V^1$ and $V^2$, and output $V$.
We then add a long chain from $V$ to $V'$, and an OR gadget with input $V'$ and $V^3$, and output $V^4$.
The clause gadget is depicted in~\cref{fig:cl-gadget}.
We call~\emph{first OR gadget} of the clause gadget, the one with output $V$, and \emph{second OR gadget}, the one with output $V^4$.

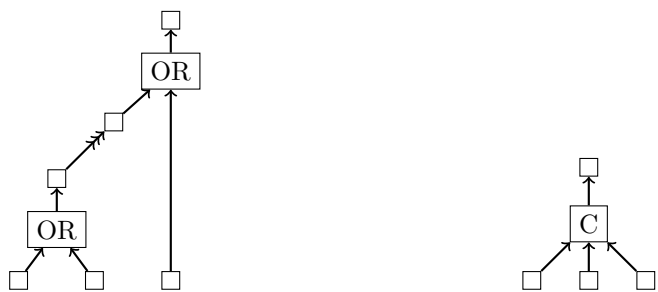
\begin{figure}[h!]
  \centering
  \begin{tikzpicture}
    \def\s{0.75}
    \begin{scope}
    \foreach \l/\i/\j in {a/0.33/0,b/1.66/0,c/1/1.8, f/2/2.8,g/3/0,i/3/4.6}{
      \node[draw] (\l) at (\s * \i, \s * \j) {} ;
    }
    \node[draw] (og) at (\s * 1, \s * 0.9) {OR} ;
    \node[draw] (h) at (\s * 3, \s * 3.7) {OR} ;
    \foreach \i/\j in {a/og,b/og,og/c,f/h,g/h,h/i}{
      \draw[thick,->] (\i) -- (\j) ;
    }
    \foreach \i/\j in {c/f}{
      \draw[thick,->>>] (\i) -- (\j) ;
    }
    \end{scope}

    \begin{scope}[xshift=7cm]
    \foreach \l/\i/\j in {a/0/0,b/1/0,c/2/0,d/1/2}{
      \node[draw] (\l) at (\s * \i, \s * \j) {} ;
    }
    \node[draw] (cl) at (\s * 1, \s * 1) {C} ;
    \foreach \i/\j in {a/cl,b/cl,c/cl,cl/d}{
      \draw[thick,->] (\i) -- (\j) ;
    }
    \end{scope}
  \end{tikzpicture}
  \caption{Left: A 3-clause gadget. Right: A shorthand for the gadget.}
  \label{fig:cl-gadget}
\end{figure}

\textbf{Constraint of the clause gadget on a 4-sequence.}
As a consequence of~\cref{lem:or1,cor:wire1}, we get that once a contraction involves the output, at least one of the vertical pairs of the inputs has been contracted.

\begin{lemma}\label{lem:clause-gadget}
  Assume $G$ contains a clause gadget with inputs $V^1, V^2, V^3$ and output $V^4$.
  In a partial 4-sequence from $G$, any contraction involving a vertex of $V^4$ is preceded by the contraction of the vertical pair of one of $V^1$, $V^2$, or $V^3$.
\end{lemma}
\begin{proof}
  Assume a contraction involves a vertex of $V^4$.
  By~\cref{lem:or1} applied to the second OR gadget, the vertical pair of $V^3$ or of $V'$ has to be contracted beforehand.
  If the vertical pair $V^3$ has been contracted, we conclude.
  Otherwise, the vertical pair of $V'$ has been contracted, and by~\cref{cor:wire1}, this implies that the vertical pair of $V$ has been contracted.
  By~\cref{lem:or1} applied to the first OR gadget, this in turn implies that the vertical pair of $V^1$ or of $V^2$ has been contracted.
\end{proof}

\textbf{Contraction of the clause gadget.}
The OR gates of the clause gadgets will be contracted as specified in~\cref{lem:or2}, while we wait the overall construction to describe the contraction of the wires. 

\subsection{Overall construction and correctness}\label{subsec:constr-corr}

Let $I=(C_1, \ldots, C_m)$ be an instance of \textsc{$3$-SAT}, that is, a collection of $m$~3-clauses over the variables $x_1, \ldots, x_n$.
We further assume that each variable appears at most twice positively, and at most twice negatively in $I$.
The \textsc{$3$-SAT} problem remains NP-complete with that restriction, and without $2^{o(n)}$-time algorithm unless the ETH fails; see~\cref{thm:3sat4occ}.
We build a trigraph $G$ that has twin-width at most~4 if and only if $I$ is satisfiable.
As trigraph $G$ will satisfy the condition of~\cref{lem:reduction-trigraph-sd}, it can be replaced by a graph on $O(|V(G)|)$ vertices.
We set $L$ the length of the long chain to $2\lceil \log(5n+3m) \rceil=O(\log n)$.

\subsubsection{Construction}

We now piece the gadgets described in the previous sections together.

\medskip

\textbf{Variable to clause gadgets.}
For every variable~$x_i$, we add a variable gadget half-guarded by $V^1_i, V^2_i$, and with outputs $V^3_i, V^4_i$.
We add a long chain starting at vertical set $V^3_i$ (resp.~$V^4_i$), and ending at a branching vertical set from which starts two long chains stopping at vertical sets $V^{\top,1}_i$ and $V^{\top,2}_i$ (resp.~$V^{\bot,1}_i$ and $V^{\bot,2}_i$).
Vertical set $V^{\top,1}_i$ (resp.~$V^{\bot,1}_i$) serves as the input of the first clause gadget in which~$x_i$ appears positively (resp.~negatively), while $V^{\top,2}_i$ (resp.~$V^{\bot,2}_i$) becomes the input of the second clause gadget in which $x_i$ appears positively (resp.~negatively).
If a literal has only one occurrence, then we omit the corresponding vertical set, and the long chain leading to it.
We nevertheless assume that each literal has at least one occurrence, otherwise the corresponding variable could be safely assigned.

\medskip

\textbf{Clause gadgets to global output.}
For every $j \in [m]$, we add a long chain from the output of the clause gadget of $C_j$, to a vertical set, denoted by $V^c_j$.
For every $j \in [2,m]$, we then add a long chain starting at $V^c_j$ and ending at $V^c_{j-1}$.
We add an arc from $V^c_1$ to a new vertical set $V^o$, which we call the \emph{global output}.

\medskip

\textbf{Global output back to half-guarding the variable gadgets.}
For every $i \in [n]$, we add two vertical sets $V^{1,r}_i, V^{2,r}_i$, and puts a long chain starting at $V^{1,r}_i$ and ending at $V^{2,r}_i$.
We add a long chain from $V^o$ to $V^{1,r}_1$.
We also add a long chain from $V^{2,r}_i$ to $V^{2,r}_{i+1}$, for every $i \in [n-1]$.
Finally we add a long chain from $V^{a,r}_i$ to $V^a_i$ for every $a \in \{1,2\}$ and every $i \in [n]$.
Recall that $V^1_i$ and $V^2_i$ are half-guarding the variable gadget of $x_i$.

This finishes the construction of the graph $G=G(I)$.
See~\cref{fig:overall} for an illustration.

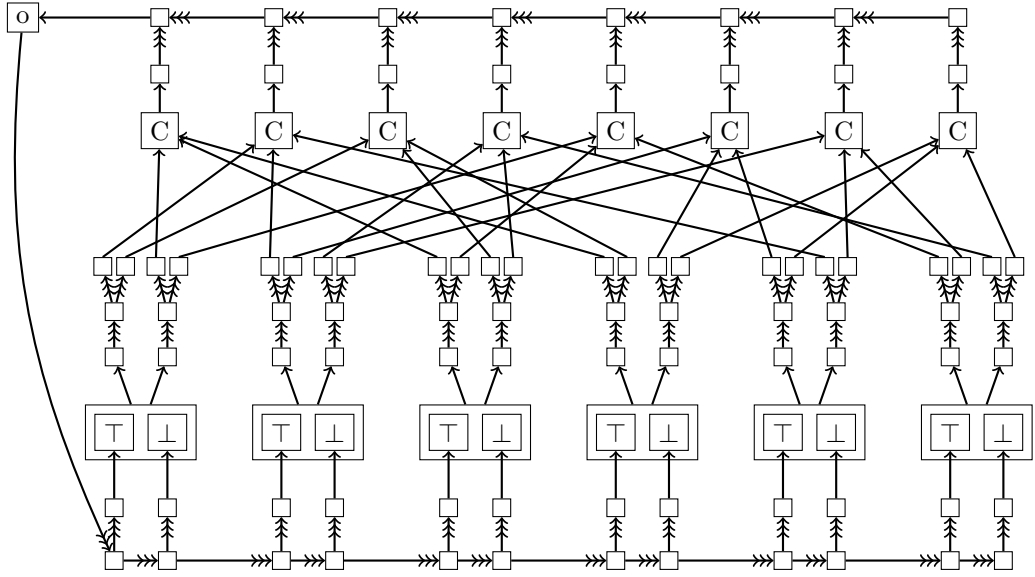
\begin{figure}
  \centering
  \begin{tikzpicture}
    \def\n{6}
    \def\m{8}
    \def\hs{2.2}
    \def\hm{1.5}
    
    \foreach \k in {1,...,\n}{
    \begin{scope}[xshift=\k * \hs cm]
      \foreach \l/\i/\j/\w in {a/0/0/, b/0.7/0/, c/0/2/, d/0.7/2/, cp/0/2.6/, dp/0.7/2.6/, t/0/1/{$\top$}, f/0.7/1/{$\bot$}, aa/0/-0.7/, bb/0.7/-0.7/, c1/-0.15/3.2/,c2/0.15/3.2/, d1/0.55/3.2/,d2/0.85/3.2/}{
        \node[draw] (\l\k) at (\i,\j) {\w} ;
      }
      \node[draw,fit=(t\k) (f\k)] (vg\k) {} ;
      \foreach \i/\j in {a/t, b/f, vg/c, vg/d}{
        \draw[thick,->] (\i\k) -- (\j\k) ;
      }
      \foreach \i/\j in {aa/a, bb/b,aa/bb,cp/c1,cp/c2,dp/d1,dp/d2,c/cp, d/dp}{
        \draw[thick,->>>] (\i\k) -- (\j\k) ;
      }
    \end{scope}
    }
    \foreach \k [count = \km from 1] in {2,...,\n}{
      \draw[thick,->>>] (bb\km) -- (aa\k) ;
    }

    \foreach \h in {1,...,\m}{
      \node[draw] (cl\h) at (1.3+\h * \hm,5) {C} ;
      \node[draw] (z\h) at (1.3+\h * \hm,5.75) {} ;
      \node[draw] (w\h) at (1.3+\h * \hm,6.5) {} ;
      \draw[thick,->] (cl\h) -- (z\h) ;
      \draw[thick,->>>] (z\h) -- (w\h) ;
    }
    \foreach \k [count = \km from 1] in {2,...,\m}{
      \draw[thick,->>>] (w\k) -- (w\km) ;
    }
    \node[draw] (o) at (1,6.5) {o} ;
    \draw[thick,->] (w1) -- (o) ;
    \draw[thick,->>>] (o) to [bend left=-15] (aa1) ;

    \foreach \i/\j in {d11/cl1,c13/cl1,c14/cl1, c11/cl2,c12/cl2,d15/cl2, c21/cl3,d13/cl3,c24/cl3, d12/cl4,d23/cl4,d16/cl4, d21/cl5,c23/cl5,c16/cl5, c22/cl6,d14/cl6,c15/cl6, d22/cl7,d25/cl7,c26/cl7, d24/cl8,c25/cl8,d26/cl8}{
      \draw[thick,->] (\i.north) -- (\j) ;
    }
  \end{tikzpicture}
  \caption{An example of the overall construction $G$ on a \textsc{$3$-SAT} instance with 6 variables and 8~clauses.
  The first two clauses are $\neg x_1 \lor x_3 \lor x_4$ and $x_1 \lor x_2 \lor \neg x_5$.}
  \label{fig:overall}
\end{figure}

\subsubsection{Correctness}

\textbf{If $G$ has twin-width at most 4, then $I$ is satisfiable.}
Let us consider the trigraph $H$ obtained after the vertical pair of the global output $V^o$ is contracted (see~top-right picture of~\cref{fig:overall-contraction}).
This has to happen in a 4-sequence by the first item of \cref{cor:fg4} applied to the fence of $V^o$.
By~\cref{cor:wire1}, no contraction involving vertices of the vertical sets $V^1_i, V^2_i$ can have happened (for any $i \in [n]$).
This is because there is a directed path in the propagation digraph $\mathcal D(G)$ from $V^o$ to $V^1_i$ and~$V^2_i$.

Thus by~\cref{lem:VG1,lem:VG3}, for every variable $x_i$ ($i \in [n]$), at most one of the wire of $+x_i$ and the wire of $\neg x_i$ has been primed.
We can define the corresponding truth assignment~$\mathcal A$: $x_i$ is set to true if the wire of $\neg x_i$ is \emph{not} primed, and to false if instead the wire of $+x_i$ is \emph{not} primed.
Besides, by~\cref{cor:wire1} applied to the contraction in $V^o$, every vertical pair of a~clause-gadget output has been contracted.
Then~\cref{lem:clause-gadget} implies that the vertical pair of at least one input of each clause gadget has been contracted.
But such a vertical pair can be contracted only if it corresponds to a literal set to true by $\mathcal A$.
For otherwise, the root of the wire of that literal cannot be contracted.
This implies that $\mathcal A$ is a satisfying assignment.

\medskip

\textbf{If $I$ is satisfiable, then $G$ has twin-width at most 4.}
In what follows, when we write that we \emph{can} contract a vertex, a set, or make a sequence of contractions, we mean that there is a partial 4-sequence that does the job. 
Let $\mathcal A$ be a satisfying assignment of~$I$.
We start by contracting ``half'' of the variable gadget of each $x_i$.
We add a subscript matching the variable index to the vertex and set labels of~\cref{fig:VG}.
For every $i \in [n]$, we contract vertices $x_i$ and $\top_i$ together, and $T_i \cup F^\top_i$ to a single vertex $v_i$, if $\mathcal A$ sets variable $x_i$ to true, and $x_i$ and $\bot_i$ together, and $U_i \cup F^\bot_i$ to a single vertex $v_i$, if instead $\mathcal A$ sets $x_i$ to false.
By~\cref{lem:VG3}, this can be done by a partial 4-sequence.

Next we contract the wire of $+x_i$ if $\mathcal A(x_i)$ is true, or the wire of $\neg x_i$ if $\mathcal A(x_i)$ is false.
This is done as described in the end of~\cref{sec:wire}.
We contract the vertical pair of the root of the wire into $z_i$ (the red degree of $v_i$ goes from 1 to 2). 
We then contract its fence by~\cref{lem:fg5}.
We can now contract the resulting vertex with $z_i$. 
Inductively, we contract the children of the current vertical set to single vertices, in a similar fashion.
As the propagation digraph has degree at most~3, this never creates a vertex of red degree more than~4.
At this point the trigraph corresponds to the top-left drawing of~\cref{fig:overall-contraction}. 

At the leaves of the wire (the vertical sets $V^{\top,1}_i, V^{\top,2}_i$ or $V^{\bot,1}_i, V^{\bot,2}_i$), we make an exception, and only contract the vertical pair.
We then contract, by \cref{lem:or2}, all the (non-already reduced) OR gadgets (involved in clause gadgets) at least one input of which has its vertical pair contracted. 
After that, applying~\cref{lem:fg5}, we finish the contraction of the OR inputs whose vertical pairs was contracted.
Next we contract each output of a~contracted OR gadget into a~single vertex.

In those clauses where the third literal is not satisfied by $\mathcal A$, the output of the clause gadget is not contracted at that point.
However, as $\mathcal A$ is a satisfying assignment, the output of the first OR gate of the gadget is contracted.
We then contract each vertical set of the long chain leading to the input of the second OR gate.
We only contract the vertical pair of that input, and contract the incident OR gadget, by~\cref{lem:or2}.
We finally proceed by contracting the input vertical set into a single vertex, and the output vertical set into a single vertex.

At this point, each output of the clause gadgets is contracted into a single vertex.
The (\emph{not} strongly) connected component $C$ of the propagation digraph $\mathcal D(G)$ containing the global output $V_o$ is acyclic and has total degree at most~3.
All the vertical sets of $C$ with in-degree 0 (ant out-degree 1) in $\mathcal D(G)$ are the clause outputs, which have been contracted to single vertices. 
Thus each vertical set of $C$ can be contracted to single vertices, by repeated use of~\cref{lem:fg5}, followed by contracting the pendant vertex (resulting from the fence contraction) with its unique (red) neighbor.

Note that this process terminates by contracting the half-guards $V_i^1, V_i^2$ (for every $i \in [n]$).
The conditions of \cref{lem:VG4} are now satisfied, so we can finish contracting each variable gadget into two vertices that we further contract together.
This results in a vertex of red degree 4.
See top-right and bottom-left drawings of~\cref{fig:overall-contraction} for the additional contractions leading to that point.

We can then contract the wire of the literal that was set to false by~$\mathcal A$.
Again, we only contract the vertical pair of the inputs of OR gates that are not contracted yet.
Then we contract those OR gadgets, and finish by fully contracting the vertical set of those inputs.
We next contract each output of the newly contracted OR gates.
We eventually contract into a single vertex each vertical set of the long chain in clause gadgets where this was not already done.

At this point, the current trigraph $H$ has only red edges (and corresponds to the bottom-right drawing of~\cref{fig:overall-contraction}).
Thus it can be interpreted as a graph, and we write~\emph{degree} instead or \emph{red degree}.
$H$ has $4n+3m$ vertices of degree 3, $n$~vertices of degree 4, and the rest of its vertices have degree~2.
Because we added long chains to separate what now corresponds to vertices of degree at least~3, $H$ is an $(\geqslant L)$-subdivision of a graph on $5n+3m$ vertices.
Since $L=2 \lceil \log(5n+3m) \rceil$, by~\cref{lem:subd}, $H$ finally admits a 4-sequence.

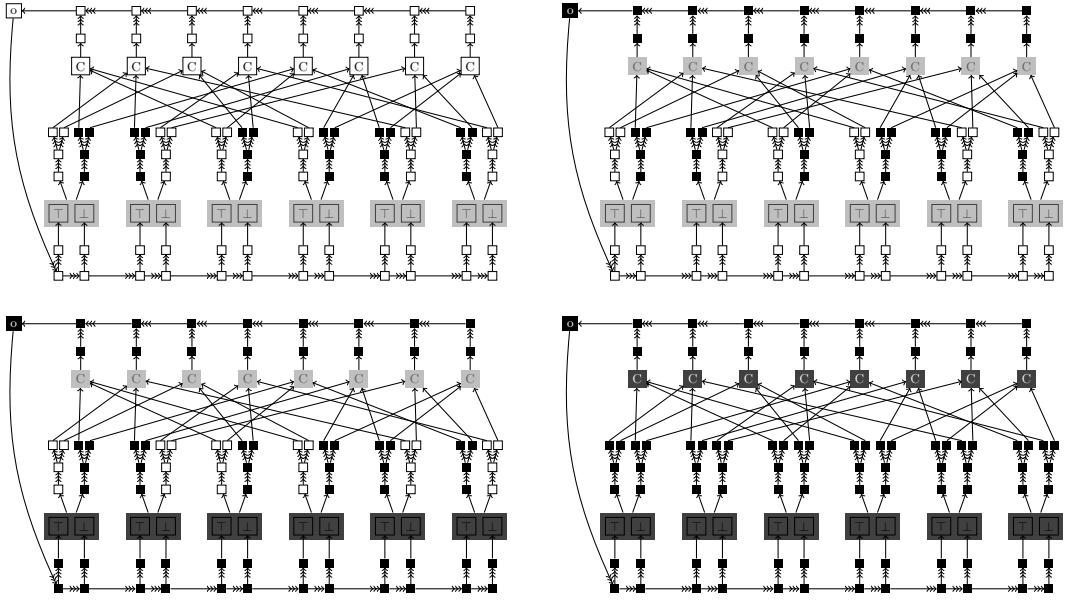
\begin{figure}
  \centering
  \resizebox{400pt}{!}{
  \begin{tikzpicture}
    \def\n{6}
    \def\m{8}
    \def\hs{2.2}
    \def\hm{1.5}

    \begin{scope}
    \foreach \k/\c/\d in {1/draw/fill,2/fill/draw,3/draw/fill,4/draw/fill,5/fill/draw,6/fill/draw}{
    \begin{scope}[xshift=\k * \hs cm]
      \foreach \l/\i/\j/\w/\x in {a/0/0//draw, b/0.7/0//draw, c/0/2//\c, d/0.7/2//\d, cp/0/2.6//\c, dp/0.7/2.6//\d, t/0/1/{$\top$}/draw, f/0.7/1/{$\bot$}/draw, aa/0/-0.7//draw, bb/0.7/-0.7//draw, c1/-0.15/3.2//\c,c2/0.15/3.2//\c, d1/0.55/3.2//\d,d2/0.85/3.2//\d}{
        \node[\x] (\l\k) at (\i,\j) {\w} ;
      }
      \node[fill opacity=0.5,fill=gray,fit=(t\k) (f\k)] (vg\k) {} ;
      \foreach \i/\j in {a/t, b/f, vg/c, vg/d}{
        \draw[thick,->] (\i\k) -- (\j\k) ;
      }
      \foreach \i/\j in {aa/a, bb/b,aa/bb,cp/c1,cp/c2,dp/d1,dp/d2,c/cp, d/dp}{
        \draw[thick,->>>] (\i\k) -- (\j\k) ;
      }
    \end{scope}
    }
    \foreach \k [count = \km from 1] in {2,...,\n}{
      \draw[thick,->>>] (bb\km) -- (aa\k) ;
    }

    \foreach \h in {1,...,\m}{
      \node[draw] (cl\h) at (1.3+\h * \hm,5) {C} ;
      \node[draw] (z\h) at (1.3+\h * \hm,5.75) {} ;
      \node[draw] (w\h) at (1.3+\h * \hm,6.5) {} ;
      \draw[thick,->] (cl\h) -- (z\h) ;
      \draw[thick,->>>] (z\h) -- (w\h) ;
    }
    \foreach \k [count = \km from 1] in {2,...,\m}{
      \draw[thick,->>>] (w\k) -- (w\km) ;
    }
    \node[draw] (o) at (1,6.5) {o} ;
    \draw[thick,->] (w1) -- (o) ;
    \draw[thick,->>>] (o) to [bend left=-15] (aa1) ;

    \foreach \i/\j in {d11/cl1,c13/cl1,c14/cl1, c11/cl2,c12/cl2,d15/cl2, c21/cl3,d13/cl3,c24/cl3, d12/cl4,d23/cl4,d16/cl4, d21/cl5,c23/cl5,c16/cl5, c22/cl6,d14/cl6,c15/cl6, d22/cl7,d25/cl7,c26/cl7, d24/cl8,c25/cl8,d26/cl8}{
      \draw[thick,->] (\i.north) -- (\j) ;
    }
    \end{scope}

    \begin{scope}[xshift=15cm]
    \foreach \k/\c/\d in {1/draw/fill,2/fill/draw,3/draw/fill,4/draw/fill,5/fill/draw,6/fill/draw}{
    \begin{scope}[xshift=\k * \hs cm]
      \foreach \l/\i/\j/\w/\x in {a/0/0//draw, b/0.7/0//draw, c/0/2//\c, d/0.7/2//\d, cp/0/2.6//\c, dp/0.7/2.6//\d, t/0/1/{$\top$}/draw, f/0.7/1/{$\bot$}/draw, aa/0/-0.7//draw, bb/0.7/-0.7//draw, c1/-0.15/3.2//\c,c2/0.15/3.2//\c, d1/0.55/3.2//\d,d2/0.85/3.2//\d}{
        \node[\x] (\l\k) at (\i,\j) {\w} ;
      }
      \node[fill opacity=0.5,fill=gray,fit=(t\k) (f\k)] (vg\k) {} ;
      \foreach \i/\j in {a/t, b/f, vg/c, vg/d}{
        \draw[thick,->] (\i\k) -- (\j\k) ;
      }
      \foreach \i/\j in {aa/a, bb/b,aa/bb,cp/c1,cp/c2,dp/d1,dp/d2,c/cp, d/dp}{
        \draw[thick,->>>] (\i\k) -- (\j\k) ;
      }
    \end{scope}
    }
    \foreach \k [count = \km from 1] in {2,...,\n}{
      \draw[thick,->>>] (bb\km) -- (aa\k) ;
    }

    \foreach \h in {1,...,\m}{
      \node[fill opacity=0.5,fill=gray] (cl\h) at (1.3+\h * \hm,5) {C} ;
      \node[fill] (z\h) at (1.3+\h * \hm,5.75) {} ;
      \node[fill] (w\h) at (1.3+\h * \hm,6.5) {} ;
      \draw[thick,->] (cl\h) -- (z\h) ;
      \draw[thick,->>>] (z\h) -- (w\h) ;
    }
    \foreach \k [count = \km from 1] in {2,...,\m}{
      \draw[thick,->>>] (w\k) -- (w\km) ;
    }
    \node[fill] (o) at (1,6.5) {o} ;
    \node at (1,6.5) {\textcolor{white}{o}} ;
    \draw[thick,->] (w1) -- (o) ;
    \draw[thick,->>>] (o) to [bend left=-15] (aa1) ;

    \foreach \i/\j in {d11/cl1,c13/cl1,c14/cl1, c11/cl2,c12/cl2,d15/cl2, c21/cl3,d13/cl3,c24/cl3, d12/cl4,d23/cl4,d16/cl4, d21/cl5,c23/cl5,c16/cl5, c22/cl6,d14/cl6,c15/cl6, d22/cl7,d25/cl7,c26/cl7, d24/cl8,c25/cl8,d26/cl8}{
      \draw[thick,->] (\i.north) -- (\j) ;
    }
    \end{scope}

    \begin{scope}[yshift=-8.5cm]
    \foreach \k/\c/\d in {1/draw/fill,2/fill/draw,3/draw/fill,4/draw/fill,5/fill/draw,6/fill/draw}{
    \begin{scope}[xshift=\k * \hs cm]
      \foreach \l/\i/\j/\w/\x in {a/0/0//fill, b/0.7/0//fill, c/0/2//\c, d/0.7/2//\d, cp/0/2.6//\c, dp/0.7/2.6//\d, t/0/1/{$\top$}/draw, f/0.7/1/{$\bot$}/draw, aa/0/-0.7//fill, bb/0.7/-0.7//fill, c1/-0.15/3.2//\c,c2/0.15/3.2//\c, d1/0.55/3.2//\d,d2/0.85/3.2//\d}{
        \node[\x] (\l\k) at (\i,\j) {\w} ;
      }
      \node[fill opacity=0.75,fill=black,fit=(t\k) (f\k)] (vg\k) {} ;
      \foreach \i/\j in {a/t, b/f, vg/c, vg/d}{
        \draw[thick,->] (\i\k) -- (\j\k) ;
      }
      \foreach \i/\j in {aa/a, bb/b,aa/bb,cp/c1,cp/c2,dp/d1,dp/d2,c/cp, d/dp}{
        \draw[thick,->>>] (\i\k) -- (\j\k) ;
      }
    \end{scope}
    }
    \foreach \k [count = \km from 1] in {2,...,\n}{
      \draw[thick,->>>] (bb\km) -- (aa\k) ;
    }

    \foreach \h in {1,...,\m}{
      \node[fill opacity=0.5,fill=gray] (cl\h) at (1.3+\h * \hm,5) {C} ;
      \node[fill] (z\h) at (1.3+\h * \hm,5.75) {} ;
      \node[fill] (w\h) at (1.3+\h * \hm,6.5) {} ;
      \draw[thick,->] (cl\h) -- (z\h) ;
      \draw[thick,->>>] (z\h) -- (w\h) ;
    }
    \foreach \k [count = \km from 1] in {2,...,\m}{
      \draw[thick,->>>] (w\k) -- (w\km) ;
    }
    \node[fill] (o) at (1,6.5) {o} ;
    \node at (1,6.5) {\textcolor{white}{o}} ;
    \draw[thick,->] (w1) -- (o) ;
    \draw[thick,->>>] (o) to [bend left=-15] (aa1) ;

    \foreach \i/\j in {d11/cl1,c13/cl1,c14/cl1, c11/cl2,c12/cl2,d15/cl2, c21/cl3,d13/cl3,c24/cl3, d12/cl4,d23/cl4,d16/cl4, d21/cl5,c23/cl5,c16/cl5, c22/cl6,d14/cl6,c15/cl6, d22/cl7,d25/cl7,c26/cl7, d24/cl8,c25/cl8,d26/cl8}{
      \draw[thick,->] (\i.north) -- (\j) ;
    }
    \end{scope}

    \begin{scope}[xshift=15cm, yshift=-8.5cm]
    \foreach \k/\c/\d in {1/fill/fill,2/fill/fill,3/fill/fill,4/fill/fill,5/fill/fill,6/fill/fill}{
    \begin{scope}[xshift=\k * \hs cm]
      \foreach \l/\i/\j/\w/\x in {a/0/0//fill, b/0.7/0//fill, c/0/2//\c, d/0.7/2//\d, cp/0/2.6//\c, dp/0.7/2.6//\d, t/0/1/{$\top$}/draw, f/0.7/1/{$\bot$}/draw, aa/0/-0.7//fill, bb/0.7/-0.7//fill, c1/-0.15/3.2//\c,c2/0.15/3.2//\c, d1/0.55/3.2//\d,d2/0.85/3.2//\d}{
        \node[\x] (\l\k) at (\i,\j) {\w} ;
      }
      \node[fill opacity=0.75,fill=black,fit=(t\k) (f\k)] (vg\k) {} ;
      \foreach \i/\j in {a/t, b/f, vg/c, vg/d}{
        \draw[thick,->] (\i\k) -- (\j\k) ;
      }
      \foreach \i/\j in {aa/a, bb/b,aa/bb,cp/c1,cp/c2,dp/d1,dp/d2,c/cp, d/dp}{
        \draw[thick,->>>] (\i\k) -- (\j\k) ;
      }
    \end{scope}
    }
    \foreach \k [count = \km from 1] in {2,...,\n}{
      \draw[thick,->>>] (bb\km) -- (aa\k) ;
    }

    \foreach \h in {1,...,\m}{
      \node[fill opacity=0.75,fill=black] (cl\h) at (1.3+\h * \hm,5) {C} ;
      \node (cl\h) at (1.3+\h * \hm,5) {\textcolor{white}{C}} ;
      \node[fill] (z\h) at (1.3+\h * \hm,5.75) {} ;
      \node[fill] (w\h) at (1.3+\h * \hm,6.5) {} ;
      \draw[thick,->] (cl\h) -- (z\h) ;
      \draw[thick,->>>] (z\h) -- (w\h) ;
    }
    \foreach \k [count = \km from 1] in {2,...,\m}{
      \draw[thick,->>>] (w\k) -- (w\km) ;
    }
    \node[fill] (o) at (1,6.5) {o} ;
    \node at (1,6.5) {\textcolor{white}{o}} ;
    \draw[thick,->] (w1) -- (o) ;
    \draw[thick,->>>] (o) to [bend left=-15] (aa1) ;

    \foreach \i/\j in {d11/cl1,c13/cl1,c14/cl1, c11/cl2,c12/cl2,d15/cl2, c21/cl3,d13/cl3,c24/cl3, d12/cl4,d23/cl4,d16/cl4, d21/cl5,c23/cl5,c16/cl5, c22/cl6,d14/cl6,c15/cl6, d22/cl7,d25/cl7,c26/cl7, d24/cl8,c25/cl8,d26/cl8}{
      \draw[thick,->] (\i.north) -- (\j) ;
    }
    \end{scope}
  \end{tikzpicture}
  }
  \caption{The different stages of the contraction sequence. Gadgets in black are contracted to single vertices, while gadgets in gray are only partially contracted.}
  \label{fig:overall-contraction}
\end{figure}

\subsubsection{Wrapping up}

The initial trigraph $G$ comprises $O((n+m)L)=O(n \log n)$ gadgets and vertical sets, each consisting of $O(1)$ vertices.
Hence $|V(G)| = O(n \log n)$.
It is immediate that the construction of $G$ can be made in polynomial time. 

The red graph of $G$ is a disjoint union of paths of length 12 (fence gadgets), and isolated vertices (the rest of the trigraph).
Thus, by~\cref{lem:reduction-trigraph-sd}, there is a \emph{graph} $G'$ on $O(|V(G)|)=O(n \log n)$ vertices such that $G'$ has twin-width at most~4 if and only if ($G$ has twin-width at most~4 if and only if) $I$ is satisfiable.
This concludes the proof of~\cref{thm:main} since a $2^{o(N/\log N)}$-time algorithm deciding if an $N$-vertex graph has twin-width at~most~4, would allow to decide \textsc{$3$-SAT} where each variable appears at most twice positively and at most twice negatively, in time $2^{o(\frac{n \log n}{\log n + \log \log n})}=2^{o(n)}$, contradicting the ETH.


\end{document}